\documentclass[11pt]{article}

\usepackage[margin=1in]{geometry}
\usepackage[utf8]{inputenc} 
\usepackage[T1]{fontenc}    
\usepackage{hyperref}       
\usepackage{url}            
\usepackage{booktabs}       
\usepackage{amssymb, amsfonts}       
\usepackage{nicefrac}       
\usepackage{microtype}      
\usepackage{xcolor}         

\title{Adaptive normalization for IPW estimation}

\author{%
	Samir Khan \\
  Stanford University\\
  \texttt{samirk@stanford.edu} \\
  \and
   Johan Ugander \\
   Stanford University \\
   \texttt{jugander@stanford.edu} \\
}

\usepackage{graphicx}
\usepackage{amsthm, amsmath}
\usepackage{optidef}
\usepackage{multicol}
\usepackage{enumerate}

\newcommand{\R}{\mathbb{R}}
\newcommand{\E}{\mathbb{E}}
\renewcommand{\P}{\mathbb{P}}
\DeclareMathOperator{\var}{var}
\DeclareMathOperator{\cov}{cov}

\newcommand{\Hajek}{\text{H\'ajek}}
\newcommand{\AN}{\text{AN}}
\newcommand{\muanaipw}{\hat{\mu}_{\text{AIPW,\AN}}}

\theoremstyle{theorem}
\newtheorem{theorem}{Theorem}
\newtheorem{lemma}{Lemma}
\theoremstyle{definition}
\newtheorem{assumption}{Assumption}
\DeclareMathOperator*{\argmin}{argmin} 

\begin{document}

\maketitle

\begin{abstract}

	Inverse probability weighting (IPW) is a general tool in survey sampling and causal inference, used both in Horvitz--Thompson estimators, which normalize by the sample size, and \Hajek{}/self-normalized estimators, which normalize by the sum of the inverse probability weights. In this work we study a family of IPW estimators, first proposed by Trotter and Tukey in the context of Monte Carlo problems, that are normalized by an affine combination of the sample size and sum of inverse weights. We show how selecting an estimator from this family in a data-dependent way to minimize asymptotic variance leads to an iterative procedure that converges to an estimator with connections to regression control methods. We refer to this estimator as an adaptively normalized estimator. For mean estimation in survey sampling, this estimator has asymptotic variance that is never worse than the Horvitz--Thompson or \Hajek{} estimators, and is smaller except in edge cases. Going further, we show that adaptive normalization can be used to propose improvements of the augmented IPW (AIPW) estimator, average treatment effect (ATE) estimators, and policy learning objectives. Appealingly, these proposals preserve both the asymptotic efficiency of AIPW and the regret bounds for policy learning with IPW objectives, and deliver consistent finite sample improvements in simulations for all three of mean estimation, ATE estimation, and policy learning.
\end{abstract}


\section{Introduction}

Consider the problem of estimating a mean from samples that are observed with non-uniform probabilities. Formally, we want to estimate the mean $\mu$ of a set of responses $Y_1,\cdots, Y_n$, but only observe $Y_1I_1,\cdots, Y_nI_n$, where $I_k$ is an indicator of whether or not unit $k$ was observed. This problem is a fundamental primitive of many problems in survey sampling, causal inference, and beyond.  We focus on the case where the $I_k$ are independent and distributed as $I_k\sim \text{Ber}(p_k)$, which can correspond to a randomized experiment under a Bernoulli design with known $p_k$ or to certain observational contexts.

The standard estimators for $\mu$ are the Horvitz--Thompson and \Hajek{} estimators~\cite{horvitz-thompson, hajek}, both of which are based on the idea of inverse probability weighting (IPW). To introduce these estimators, we first define $$\hat{S}=\sum_{k=1}^n \frac{Y_kI_k}{p_k}\quad\text{and}\quad \hat{n}=\sum_{k=1}^n \frac{I_k}{p_k}$$
as estimates of the population total and sample size. Then, the Horvitz--Thompson and \Hajek{} estimators are 
$$\hat{\mu}_{\text{HT}}=\hat{S}/n \quad\text{and}\quad \hat{\mu}_{\text{\Hajek{}}}=\hat{S}/\hat{n}.$$
These estimators have several desirable properties: $\hat{\mu}_{\text{HT}}$ is unbiased and admissible in the class of all unbiased estimators \cite{admissibility}, while $\hat{\mu}_{\text{\Hajek{}}}$ is approximately unbiased and often has lower variance than $\hat{\mu}_{\text{HT}}$ \cite{sarndal}. 

The idea of inverse probability weighting also figures prominently in the Monte Carlo literature on importance sampling, where the Horvitz--Thompson and \Hajek{} estimators are known as the importance sampling (IS) and self-normalized importance sampling (SNIS) estimators, respectively. In 1954, working in this context, Trotter and Tukey \cite{trottertukey} briefly entertained the idea of a family of estimators that generalizes $\hat{\mu}_{\text{HT}}$ and $\hat{\mu}_{\text{\Hajek{}}}$, namely 
\begin{eqnarray}
\hat{\mu}_{\lambda}=\frac{\hat{S}}{(1-\lambda)n+\lambda\hat{n}},
\label{eqn:tt}
\end{eqnarray}
for $\lambda \in \R$. This family contains Horvitz--Thompson as the special case $\lambda=0$ and \Hajek{} as the special case $\lambda=1$. Curiously, Trotter and Tukey emphasized that $\lambda$ need not be constrained to $[0,1]$, and that values outside that range might sometimes be useful \cite{trottertukey}.

Although this proposal first appeared nearly 70 years ago, it has not received any significant attention in the literature. In this work, we consider this proposal in detail. When $\lambda$ is selected as a function of the data to optimize some criterion, we refer to this idea as {\it adaptive normalization}. We show in Section~\ref{subsec:lambdastar} how choosing a member of the family \eqref{eqn:tt} in a data-dependent way---to iteratively minimize asymptotic variance---leads to estimators that improve on both the Horvitz–Thompson and \Hajek{} estimators in terms of asymptotic variance, despite the additional variance incurred by estimating $\lambda$ from the data. We also show how the adaptively normalized estimator can be understood as a regression control/control variate method in Section~\ref{subsec:connections}.

As applications, we deploy adaptive normalization in diverse settings where the Horvitz--Thompson estimator and \Hajek{} estimator are traditionally used. The popular AIPW estimator of \cite{aipw} essentially uses Horvitz--Thompson as a primitive, and we replace the sample size normalization with adaptive normalization to obtain an estimator that preserves the asymptotic efficiency of the AIPW estimator while improving performance in simulations. Similarly, the problem of estimating average treatment effects (ATEs) can be tackled by separately estimating the mean response in the treated group and the mean response in the control group using adaptive normalization. Finally, we also show that a new objective for policy learning based on adaptive normalization preserves the regret bounds of \cite{kitagawa} and learns better policies in simulations. More broadly, our analysis strongly suggests that choosing normalizations other than $n$ or $\hat n$ is worthy of attention in other contexts where inverse probability weighted estimators appear, such as off-policy evaluation or reinforcement learning.

\paragraph{Motivation for adaptive normalization.}

Why would we want to use values of $\lambda$ learned from the data in \eqref{eqn:tt} rather than $\lambda=0$ or $\lambda=1$, especially values of $\lambda<0$ or $\lambda>1$? As a starting point, consider $\lambda=0$, the Horvitz--Thompson estimator, and $\lambda=1$, the \Hajek{} estimator. As mentioned above, the \Hajek{} estimator often has lower variance than the Horvitz--Thompson estimator. One reason for this variance reduction is that the numerator and denominator of the \Hajek{} estimator are positively correlated, while the numerator and denominator of the Horvitz--Thompson estimator are uncorrelated. 

To see the role this correlation plays, consider a toy example with $n=10$ units with responses and response probabilities 
$$Y_1=Y_2=\cdots=Y_{10}=1,\quad \quad p_1=10^{-5}, p_2=\cdots = p_{10}=0.5.$$
 Then, we claim the Horvitz--Thompson and \Hajek{} estimators have very different behaviors on the event that $Y_1$ is observed. For our illustration, assume units 2--5 are observed but 6--10 are not.

For the Horvitz--Thompson estimator, if $Y_1$ is observed then the numerator of the estimator increases from 8 to $8+1/10^{-5} \approx 10^5$, while the denominator, $n$, remains fixed at 10, so our estimate increases from $8/10$ to $(8+10^5)/10\approx 10^4$. For the \Hajek{} estimator, if $Y_1$ is observed, the numerator similarly increases by $10^5$, but now the denominator also increases by $10^5$, and so our estimate is less affected. The positive correlation between the numerator and denominator of the \Hajek{} estimator provides a kind of shrinkage that serves to reduce variance.

Since the correlation between numerator and denominator in the \Hajek{} estimator is one way that it reduces variance, we would expect that we can further reduce variance by introducing more correlation, which motivates taking $\lambda>1$ in $\eqref{eqn:tt}$. By increasing the weight on the random factor of $\hat{n}$ in the denominator, we increase the covariance between the numerator and denominator, and thus reduce variance. Put crudely, if going from $\lambda =0$ to $\lambda=1$ reduces variance, shouldn't going from $\lambda=1$ to $\lambda=2$ reduce it even more?

Of course, there is a bias-variance trade-off. Increasing $\lambda$ decreases the variance of $\hat{\mu}_{\lambda}$, but increasing $\lambda$ to be arbitrarily large also shrinks our estimates towards 0 and increases bias. Thus the problem becomes one of choosing a value of $\lambda$ that provides the correct amount of shrinkage for a given problem. This trade-off is visualized in Figure \ref{fig:optim}, which shows the MSE of \eqref{eqn:tt} for a simulated problem at a range of values of $\lambda$. Ideally, we would choose $\lambda$ to minimize the curve shown in this figure; unfortunately, we cannot compute the exact MSE of $\hat{\mu}_{\lambda}$. In this work, we take the approach of showing that $\hat{\mu}_{\lambda}$ is asymptotically normal and then choose $\lambda$ to iteratively minimize estimates of the asymptotic variance. Other procedures for selecting $\lambda$ based on different criteria would lead to other estimators, suggesting a general class of adaptively normalized estimators that may be worthy of further study.

\begin{figure}[t]
	\centering
	\includegraphics[width=0.55\linewidth]{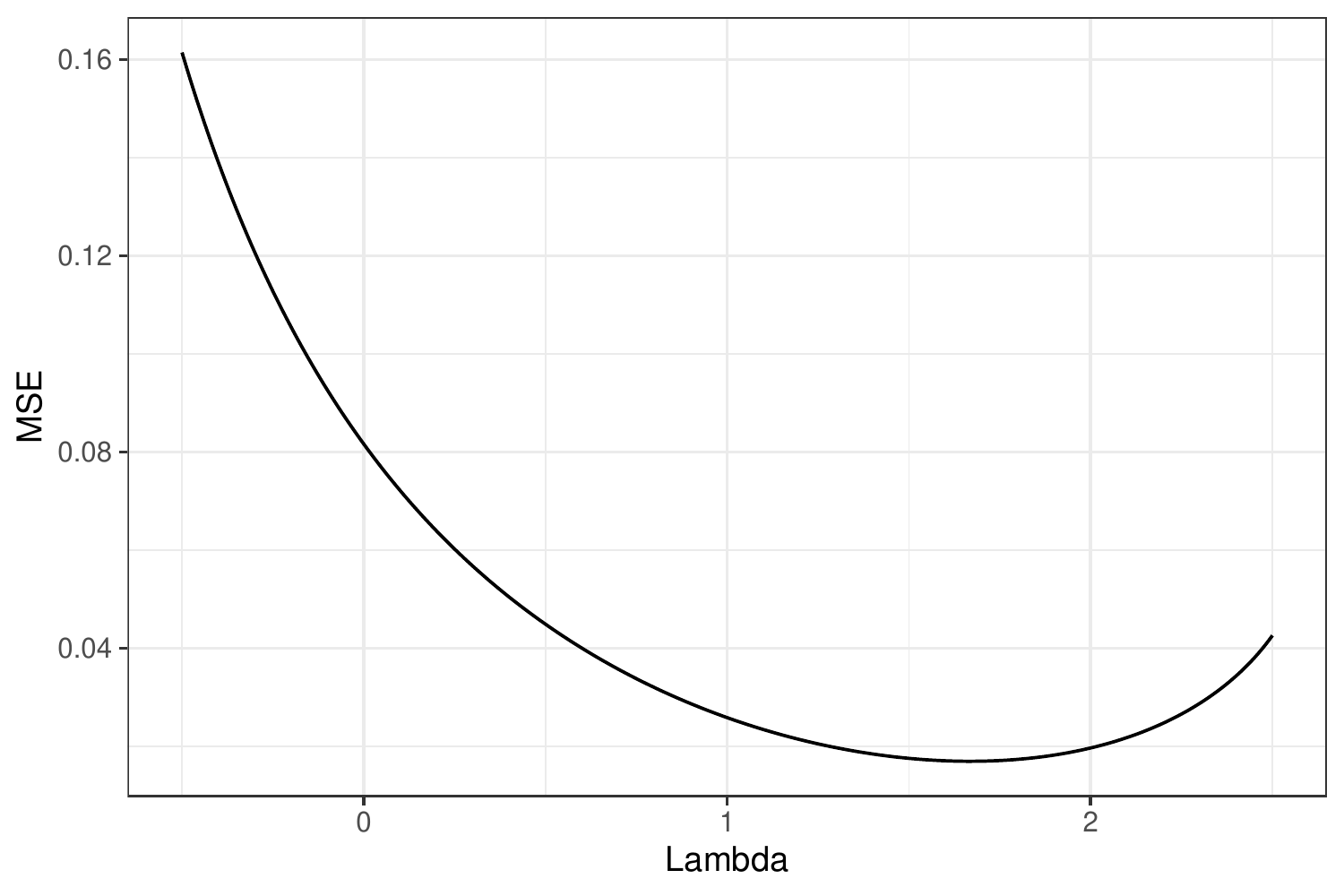}
	\caption{The MSE of $\hat{\mu}_{\lambda}$ as a function of $\lambda$ in an instance of the normal model described in \eqref{eq:normal-model} with $n=250$, $\mu=1$, and $\rho=0.2$. Note that neither the Horvitz--Thompson ($\lambda=0$) nor \Hajek{} ($\lambda=1$) estimator minimizes the MSE. Furthermore, the optimal choice of $\lambda$ is greater than 1.}
	\label{fig:optim}
\end{figure}

\subsection{Related work}

The present work is most closely related to the literature on variance reduction in importance sampling. Since the proposal of importance sampling in \cite{is}, there have been a variety of proposals for variance reduction methods, many of which are surveyed in \cite{artnotes}. The one most relevant to our work is the method of control variates, which is also known as difference estimation in the survey sampling community \cite{sarndal}. We show in Section \ref{subsec:connections} that the estimator we derive in Section \ref{sec:adapt} is algebraically equivalent to a particular kind of control variate estimator, and a slightly modified form of this estimator has been previously discussed by \cite{firth} and studied in Monte Carlo simulations by \cite{hesterberg}. For proponents of this choice of control variate, adaptive normalization provides an independent motivation and derivation. 

Our proposed modifications of the AIPW estimator stand alongside many other proposals: for example, \cite{el-aipw, el-aipw-2} use an empirical likelihood approach, while \cite{improved-aipw} uses a weighting approach, among others.  We do not, however, know of any attempt or alternative derivation that is equivalent to the estimator we derive in this context. Similarly, our policy learning proposals build directly on the work of \cite{kitagawa}, which is itself closely related to work on counterfactual risk minimization in the bandit setting \cite{bandit}, and was extended by \cite{wager} using ideas from AIPW estimation.


\section{Problem formulation and notation}

In this section we formally specify our model and introduce notation. We assume that pairs $(Y_1, p_1),\ldots, (Y_n, p_n)$ are drawn i.i.d.\ from a super-population distribution $\mathcal{D}$ on $\R\times [0,1]$ and that we observe both $p_1,\ldots, p_n$ and $Y_1I_1,\ldots, Y_nI_n$ where the $I_k$ are independent and $I_k\mid p_k\sim\text{Ber}(p_k)$.

This model differs from previous design-based work on asymptotics of the Horvitz--Thompson and \Hajek{} estimators in assuming a super-population model for the $Y_k$ rather than modeling $Y_k$ as a sequence of fixed finite populations \cite{savje-consistency, robinson-consistency}. However, we believe our results are still true in a finite population model, with slightly modified proofs, and we choose to work in a super-population model mainly to avoid making many cumbersome assumptions about the existence of limiting moments of the $Y_k$. In particular, although we are working in a super-population model, we make no parametric assumptions on $\mathcal{D}$, such as assuming a regression model.

Our assumption that the $p_k$ are random as well is best understood in the context of a more general model that we consider in Sections~\ref{subsec:covariate} and~\ref{subsec:policy}. There, we consider pairs $(Y_1, X_1),\ldots, (Y_n, X_n)$ drawn i.i.d.\ from a super-population distribution $\mathcal{D}$ on $\R\times \mathcal{X}$, where the $Y_k$ are responses and the $X_k$ are known covariates. Then, we can think of the treatment probability $p_k$ as a function $p_k=p(X_k)$ of the observed covariates. 

We typically assume that the $p_k$ (or equivalently the map $p(\cdot)$) are known exactly, rather than estimated from the $X_k$. This is mainly to facilitate asymptotic comparisons between our estimators that are independent of the model used to estimate $p_k$, and we believe that it is reasonable to use estimated propensities with our proposed estimators. In Section~\ref{subsec:covariate}, where we discuss AIPW estimation, under appropriate consistency and rate conditions, the choice of model for $p_k$ no longer plays any role in the asymptotics, and so we relax this assumption there and allow the $p_k$ to be estimated propensities.

We now specify our assumptions on $\mathcal{D}$. 

\begin{assumption}[Boundedness and overlap]
	There exist constants $M, \delta>0$ such that $|Y_k|\leq M$ and $\delta \leq p_k\leq 1-\delta$ almost surely. 
	\label{bounded}
\end{assumption}

Under Assumption \ref{bounded}, all of the moments 
\begin{eqnarray}
\mu=\E[Y_1],\quad \pi=\E\left[ \frac{1-p_1}{p_1} \right],\quad T=\E\left[ Y_1\frac{1-p_1}{p_1} \right]
\label{eqn:moments}
\end{eqnarray}
of $\mathcal{D}$ are all finite. 

Throughout the next section, our goal is to estimate $\mu=\E[Y_k]$ in the model presented here. In Section~\ref{sec:apps}, we consider more diverse models and goals.


\section{Adaptive normalization}
\label{sec:adapt}

In this section we propose and analyze a novel procedure for selecting a data-dependent value of $\lambda$ and describe properties of the adaptively normalized IPW estimator that follows. We establish that the resulting estimator has smaller asymptotic variance than either the Horvitz--Thompson or \Hajek{} estimator.

\subsection{The optimal choice of $\lambda$}

Before we can understand how to choose $\lambda$ from the data, we must first understand the behavior of $\hat{\mu}_{\lambda}$ for a fixed $\lambda$, as a function of $\lambda$ and the problem parameters. To this end, we contribute the following central limit theorem, which is a straightforward generalization of known results on the bias and variance of the \Hajek{} estimator. We defer the proof this result, as well as all other results in this section, to Appendix~\ref{sec:proofs}.
\begin{theorem}
	Suppose Assumption \ref{bounded} holds. Then, for any fixed $\lambda\in \R$, we have the CLT 
	\label{thm:mu-hat-lambda-clt}
	\begin{equation}
		\sqrt{n}(\hat{\mu}_{\lambda}-\mu)\xrightarrow{d} N\left( 0,  \sigma^2_{\lambda}\right),\quad \sigma^2_{\lambda}=\E\left[ \frac{1-p_k}{p_k}(Y_k-\lambda\mu)^2 \right].\label{eq:mu-hat-lambda-clt}
	\end{equation}
\end{theorem}
With this result in hand, we now choose $\lambda$ to minimize $\sigma^2_{\lambda}$. Moving forward, we assume that $\mu\neq 0$, since if $\mu=0$, then $\sigma_{\lambda}$ does not actually depend on $\lambda$, and minimizing over $\lambda$ is no longer meaningful. 

Under the assumption that $\mu\neq 0$ then, there is a unique value of $\lambda$ that minimizes $\sigma_{\lambda}^2$ in \eqref{eq:mu-hat-lambda-clt}, given by 
\begin{equation}
	\lambda^*=\frac{\E\left[ \frac{1-p_k}{p_k}Y_k \right]}{\mu\E\left[ \frac{1-p_k}{p_k} \right]}=\frac{T}{\pi \mu},
	\label{eq:lambda-star}
\end{equation}
where $\pi$ and $T$ are the moments defined in \eqref{eqn:moments}. Assumption \ref{bounded} precludes the possibility that $\pi=0$. 

We can interpret \eqref{eq:lambda-star} to shed light on the role $\lambda$ plays. If the $Y_k$ and $p_k$ are positively correlated, then $Y_k$ and $\frac{1-p_k}{p_k}$ are negatively correlated, so $T<\mu\pi$ and $\lambda^*<1$. Similarly, if the $Y_k$ and $p_k$ are negatively correlated, we obtain $\lambda^*>1$. This interpretation of \eqref{eq:lambda-star} extends the conventional wisdom that \Hajek{} is preferable to Horvitz--Thompson when $Y_k$ and $p_k$ are negatively correlated \cite{sarndal}.

\subsection{Estimating $\lambda^*$ from the data}
\label{subsec:lambdastar}

Based on the above asymptotic analysis, we would seem to prefer $\hat{\mu}_{\lambda^*}$ over $\hat{\mu}_{\text{HT}}$ and $\hat{\mu}_{\text{\Hajek{}}}$. However, we cannot use $\hat{\mu}_{\lambda^*}$ directly as an estimator of $\mu$ because the prescribed choice of $\lambda^*$ depends on unknown moments of $\mathcal{D}$, including the very mean $\mu$ we are trying to estimate. What happens if we estimate $\lambda^*$ from data?

Our expression of $\lambda^*$ depends on three moments of $\mathcal{D}$, namely $T, \pi, \mu$. The first two of these can readily be estimated from the data by the IPW-style estimators 

\begin{equation}
	\hat{T}=\frac{1}{n}\sum_{k=0}^n \frac{1-p_k}{p_k}Y_k\frac{I_k}{p_k},\quad \hat{\pi}=\frac{1}{n}\sum_{k=0}^n \frac{1-p_k}{p_k}\frac{I_k}{p_k},
	\label{eq:sample-moments}
\end{equation}
and we already have the estimator $\hat{\mu}_{\text{HT}}$ of $\mu$. 

As a brief aside, we might wonder whether we can also use an adaptive normalization when estimating $T$ and $\pi$, instead of a traditional IPW estimator. Unfortunately, just as adaptively normalizing an estimator of $\mu$ requires estimating the higher-order moments $T$ and $\pi$, adaptively normalizing an estimator of $T$ or $\pi$ would require estimating even further higher-order moments. To avoid these complications, we restrict ourselves to using the Horvitz--Thompson normalization. 

Now, the estimators in \eqref{eq:sample-moments} lead us to estimate $\lambda^*$ and $\mu$ by 
\begin{equation}
\hat{\lambda}^*=\frac{\hat{T}}{\hat{\pi} \hat{\mu}_{\text{HT}}},\quad \hat{\mu}_{\hat{\lambda}^*}=\frac{\hat{S}}{(1-\hat{\lambda}^*)n+\hat{\lambda}^*\hat{n}}.
\label{eq:mu-hat-1}
\end{equation}
Unfortunately, we can check in simulation that this estimator is sometimes even worse than $\hat{\mu}_{\text{HT}}$. Essentially, the additional variance introduced by estimating $\lambda^*$ from data can outpace the variance reduction from using a tuned value of $\lambda^*$. However, we can make useful observation: $\hat{\mu}_{\hat{\lambda}^*}$ is often a better estimator of $\mu$ than $\hat{\mu}_{\text{HT}}$, so should we not use $\hat{\mu}_{\hat{\lambda}^*}$ rather than $\hat{\mu}_{\text{HT}}$ when estimating $\lambda$? In fact, every time we obtain a better estimate of $\mu$, we can use this to obtain a better estimate of $\lambda^*$. On the other hand, a better estimate of $\lambda^*$ will, because $\lambda^*$ is the optimal amount of normalization, lead to a better estimate of $\mu$. Combining these ideas leads to the following alternating scheme. 

Formally, we construct a sequence of estimators $(\hat{\lambda}^{(t)}, \hat{\mu}^{(t)})$ initialized at $(\hat{\lambda}^{(0)}, \hat{\mu}^{(0)})=(0, \hat{\mu}_{\text{HT}})$ and defined for $t>0$ by the recursions

\begin{equation}
	\hat{\lambda}^{(t)}=\frac{\hat{T}}{\hat{\pi}\hat{\mu}^{(t-1)}},\quad \hat{\mu}^{(t)}=\frac{\hat{S}}{(1-\hat{\lambda}^{(t)})n+\hat{\lambda}^{(t)}\hat{n}}.
	\label{eq:updates}
\end{equation}
The first equation in (\ref{eq:updates}) corresponds to estimating $\lambda^*$ using $\hat{\mu}^{(t-1)}$ as an estimate of $\mu$, while the second corresponds to estimating $\mu$ using $\hat{\lambda}^{(t)}$ as an estimate of $\lambda^*$. There are two possible stable limiting behaviors for this sequence: the first is to trivially have $\hat{\mu}^{(t)}\to 0$ and $\hat{\lambda}^{(t)}\to \infty$, while the second is to converge to a fixed point at a pair $(\hat{\mu}_{\AN}, \hat{\lambda}_{\AN})$ satisfying $$\hat{\lambda}_{\AN}=\frac{\hat{T}}{\hat{\pi}\hat{\mu}_{\AN}},\quad \hat{\mu}_{\AN}=\frac{\hat{S}}{(1-\hat{\lambda}_{\AN})n+\hat{\lambda}_{\AN}\hat{n}}.$$ This system of equations has the unique solution 

\begin{equation}
	\hat{\mu}_{\AN}=\frac{\hat{S}}{n}+\frac{\hat{T}}{\hat{\pi}}\left( 1-\frac{\hat{n}}{n} \right).
	\label{eq:mu-hat-fp}
\end{equation}

The following theorem formally establishes that the iterations (\ref{eq:updates}) converge to the non-trivial solution, the fixed point at (\ref{eq:mu-hat-fp}).

\begin{theorem}
	\label{thm:denom-convergence}
	Suppose Assumption \ref{bounded} holds, we have $\mu\neq0$, and consider the sequence of estimators $(\hat{\lambda}^{(t)}, \hat{\mu}^{(t)})$ initialized at $\hat{\lambda}^{(0)}=0, \hat{\mu}^{(0)}=\hat{\mu}_{\text{HT}}$ and defined for $t>0$ by the recursion (\ref{eq:updates}). Then 
	\begin{enumerate}[(i)]
		\item the sequence $\hat{\mu}^{(t)}$ converges as $t\to \infty$ to an estimator $\hat{\mu}_{\text{lim}}$;
		\item the estimator $\hat{\mu}_{\text{lim}}$ satisfies $$\lim_{n\to \infty} \P\left( \hat{\mu}_{\text{lim}}=\hat{\mu}_{\AN} \right)=1,$$ so that $\hat{\mu}_{\text{lim}}-\hat{\mu}_{\AN}$ converges in probability to 0.
	\end{enumerate}
\end{theorem}

Thus, our attempts to develop an estimator of $\mu$ by estimating the optimal normalization parameter $\lambda^*$  culminate in the estimator $\hat{\mu}_{\AN}$, which we refer to as the adaptively normalized IPW estimator. The role of Theorem~\ref{thm:denom-convergence} and the iterative scheme we have presented is to make explicit the connection between adaptive normalization and the final estimator $\hat{\mu}_{\AN}$. The fixed point equations alone do not characterize this connection, since, based on those equations alone, we may be led to take $\hat{\mu}=0$ and think of this as using an infinite value of $\lambda$. Our theorem shows that the iterations \eqref{eq:updates}, which correspond to iteratively learning better and better normalizations, uniquely pick out $\hat{\mu}_{\AN}$ with high probability.

Before proceeding, we note a very attractive property of $\hat{\mu}_{\AN}$ in \eqref{eq:mu-hat-fp}: its simplicity. It is, for practical purposes, the Horvitz--Thompson estimator with a correction term, and can replace the Horvitz--Thompson and \Hajek{} estimators in essentially any application where they are used, with minimal additional computation. Beyond this simplicity, we will see in Section \ref{subsec:variance} that $\hat{\mu}_{\AN}$ typically has lower asymptotic variance than either the Horvitz--Thompson estimator or the \Hajek{} estimator, and in Section \ref{subsec:connections} that it is closely related to several existing ideas in the literature.

\paragraph{An optimization perspective.}
It may feel unsettled to minimize the asymptotic variance and then commence an iteration scheme. We can also derive $\hat{\mu}_{\AN}$ more directly by framing the problem of selecting $\lambda$ as a joint minimization, over $\mu$ and $\lambda$, of the asymptotic variance. Minimizing the asymptotic variance $\sigma^2_{\lambda}$ in (\ref{eq:mu-hat-lambda-clt}) is equivalent to minimizing $$-2\lambda\mu\E\left[ \frac{1-p_k}{p_k}Y_k \right]+\lambda^2\mu^2\E\left[ \frac{1-p_k}{p_k} \right]=-2\lambda\mu T+\lambda^2\mu^2\pi.$$ As before, we use $\hat{T}$ and $\hat{\pi}$ defined in (\ref{eq:sample-moments}) to estimate $T$ and $\pi$ but now we estimate $\mu$ by $\hat{\mu}_{\lambda}$ directly. This leads to the non-convex optimization problem $$\min_{\lambda} -2\lambda\hat{\mu}_{\lambda}\hat{T}+\lambda^2\hat{\mu}_{\lambda}^2\hat{\pi}.$$ 

Rather than directly solving this optimization problem, we consider the equivalent problem 
\begin{mini}|s|
	{\lambda, \hat{\mu}}{-2\lambda\hat{\mu}\hat{T}+\lambda^2\hat{\mu}^2\hat{\pi},}
	{}{\label{constrained-opt}}
	\addConstraint{\hat{\mu}=\frac{\hat{S}}{(1-\lambda)\hat{n}+\lambda n}}.
\end{mini}
This problem is also non-convex, but is more amenable to analysis. In particular, we claim that, despite its non-convexity, \eqref{constrained-opt} can be easily solved analytically to find a unique optimum. 

To solve \eqref{constrained-opt}, note that the objective is a quadratic function of $\lambda\hat{\mu}$ and so checking first-order stationarity conditions shows that the unconstrained minimum is achieved for any pair $(\lambda, \hat{\mu})$ satisfying $\lambda\hat{\mu}=\hat{T}/\hat{\pi}$. Then, direct algebra yields that there is a unique pair $(\hat{\lambda}_{\text{opt}}, \hat{\mu}_{\text{opt}})$ satisfying both $\hat{\lambda}_{\text{opt}}\hat{\mu}_{\text{opt}}=\hat{T}/\hat{\pi}$ as well as the constraint in (\ref{constrained-opt}), and the resulting $\hat{\mu}_{\text{opt}}$ is precisely $\hat{\mu}_{\AN}$.

\subsection{Properties of adaptive normalization}
\label{subsec:variance}

We now study the asymptotic behavior of $\hat{\mu}_{\text{AN}}$, and show that its asymptotic variance generically improves on the asymptotic variance of the \Hajek{} and Horvitz--Thompson estimators. 

\paragraph{Asymptotic variance.} 
To understand the asymptotics of $$\hat{\mu}_{\AN}=\frac{\hat{S}}{n}+\frac{\hat{T}}{\hat{\pi}}\left( 1-\frac{\hat{n}}{n} \right),$$ we use the fact that each of $\hat{S}/n, \hat{T}, \hat{\pi}, \hat{n}/n$ is an average of i.i.d.\ terms. This structure implies that the vector $(\hat{S}/n, \hat{T}, \hat{\pi}, \hat{n}/n)$ satisfies a CLT, and we can then apply the delta method to obtain a CLT for $\hat{\mu}_{\AN}$. This argument produces the following result.
\begin{theorem}
	\label{thm:mu-hat-fp-clt}
	Under Assumption \ref{bounded}, the estimator $\hat{\mu}_{\AN}$ satisfies the CLT
	\begin{equation}
		\sqrt{n}(\hat{\mu}_{\AN}-\mu)\xrightarrow{d} N\left( 0, \E\left[\frac{1-p_k}{p_k}\left(Y_k-\frac{T}{\pi}\right)^2\right]\right).
		\label{eq:mu-hat-fp-clt}
	\end{equation}
Furthermore, the asymptotic variance in \eqref{eq:mu-hat-fp-clt} is always smaller than the asymptotic variances of $\hat{\mu}_{\text{HT}}$ and $\hat{\mu}_{\Hajek{}}$, and is strictly smaller except if $T=\mu\pi$ or $T=0$.
\end{theorem}
We defer the proof of the CLT to Appendix~\ref{sec:proofs}, but give the proof of the variance comparison below.

\begin{proof}
We begin by expanding $$\E\left[ \frac{1-p_k}{p_k}\left( Y_k-\frac{T}{\pi} \right)^2 \right]=\E\left[ Y_k^2\frac{1-p_k}{p_k} \right]-\frac{1}{\pi}\E\left[ Y_k\frac{1-p_k}{p_k} \right]^2.$$

Then, we first compare with $\hat{\mu}_{\text{HT}}$. For $\hat{\mu}_{\text{HT}}$, plugging $\lambda=0$ into \eqref{eq:mu-hat-lambda-clt} gives an asymptotic variance of $\E\left[ Y_k^2\frac{1-p_k}{p_k} \right]$. Since $\frac{1}{\pi}\E\left[ Y_k\frac{1-p_k}{p_k} \right]^2\geq 0$, this is always larger than the asymptotic variance in (\ref{eq:mu-hat-fp-clt}), and is strictly larger whenever $\E\left[ Y_k\frac{1-p_k}{p_k} \right]\neq 0$.

	For $\hat{\mu}_{\text{\Hajek{}}}$, plugging $\lambda=1$ into (\ref{eq:mu-hat-lambda-clt}) gives an asymptotic variance of $\E\left[ \frac{1-p_k}{p_k}(Y_k-\mu)^2 \right]$. Algebraic manipulations show that $$\E\left[ \frac{1-p_k}{p_k}(Y_k-\mu)^2 \right]\leq \E\left[ Y_k^2\frac{1-p_k}{p_k} \right]-\frac{1}{\pi}\E\left[ Y_k\frac{1-p_k}{p_k} \right]^2\iff (\mu\pi-T)^2\leq 0,$$ so the asymptotic variance of $\hat{\mu}_{\AN}$ is always smaller than the asymptotic variance of $\hat{\mu}_{\text{\Hajek{}}}$ unless $T=\pi \mu$, in which case they are equal. 
\end{proof}

Theorem~\ref{thm:mu-hat-fp-clt} shows that $\hat{\mu}_{\AN}$ can only ever improve on $\hat{\mu}_{\text{HT}}$ and $\hat{\mu}_{\Hajek{}}$ asymptotically. For a more careful understanding of when we expect to see improvement, note that the condition $T=\pi\mu$ is equivalent to $\E\left[ \frac{1-p_k}{p_k}Y_k \right]=\E\left[ \frac{1-p_k}{p_k} \right]\E[Y_k]$, which is equivalent to $Y_k$ and $\frac{1-p_k}{p_k}$ being uncorrelated. This condition will be satisfied, for example, when $Y_k$ and $p_k$ are independent. So in some sense, $\hat{\mu}_{\lambda}$ is exploiting dependence between $Y_k$ and $p_k$ to improve performance, and cannot improve on \Hajek{} in the absence of such dependence. 

Another useful observation is that the condition $T=\mu\pi$ corresponds to $\lambda^*=1$, while $T=0$ corresponds to $\lambda^*=0$. Thus we see that the only cases where wee fail to improve on $\hat{\mu}_{\Hajek{}}$ and $\hat{\mu}_{\text{HT}}$ are those in which $\hat{\mu}_{\Hajek{}}$ and $\hat{\mu}_{\text{HT}}$ were coincidentally using the optimal value of $\lambda$ already.

\paragraph{Finite sample variance.}

Examining the form of $\hat{\mu}_{\AN}$ directly, we can see that if it has lower variance than $\hat{\mu}_{\text{HT}}$ in finite samples, it must be because the correction term $\frac{\hat{T}}{\hat{\pi}}\left( 1-\frac{\hat{n}}{n} \right)$ is negatively correlated with the first term, $\hat{S}/n=\hat{\mu}_{\text{HT}}$. However, using the iterative scheme introduced in (\ref{eq:updates}), we can give an alternative explanation for why the finite sample variance of $\hat{\mu}_{\AN}$ is smaller than that of $\hat{\mu}_{\text{HT}}$, one that builds on the result of Theorem \ref{thm:denom-convergence}. This explanation is presented in Appendix~\ref{app:subsec:proof2}, following the proof of Theorem~\ref{thm:denom-convergence}, and suggests that under Assumption~\ref{bounded} and assuming $\mu\neq 0$, every two steps of the iterations in \eqref{eq:updates} reduce finite-sample variance. This further underscores the importance of adaptive normalization as a variance reduction technique; an experimental demonstration of this phenomenon is shown in Figure~\ref{fig:iterative-variance}.

\begin{figure}[t]
	\centering
	\includegraphics[width=0.6\linewidth]{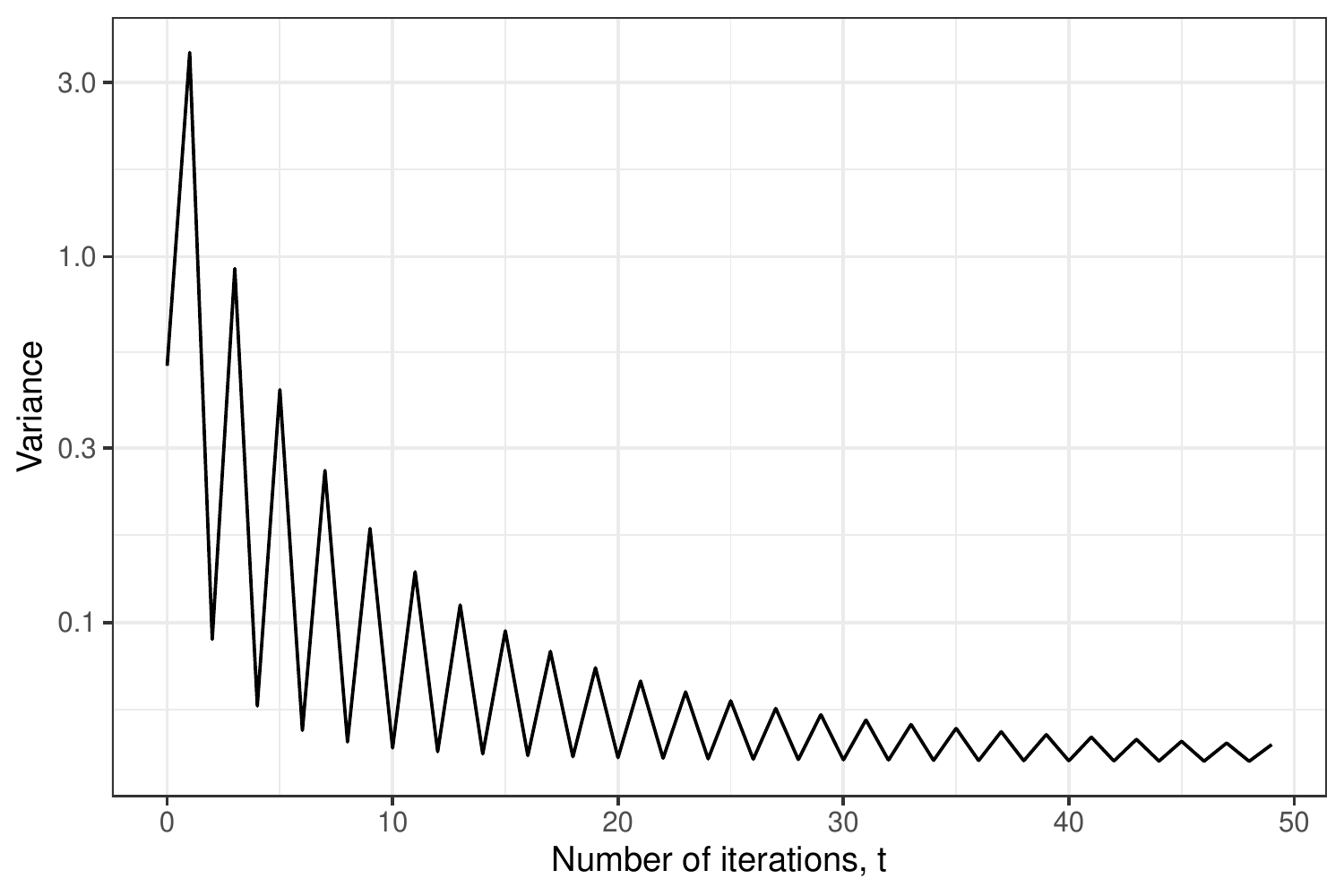}
	\caption{The variance of $\hat{\mu}^{(t)}$ as a function of $t$. Although a single iteration may increase the variance, we observe that, in this simulation, every two iterations reduce variance. The data here are generated from the normal model \eqref{eq:normal-model} in Section~\ref{sec:sims} with $n=100, \rho =0.1, \mu=-1$.}
	\label{fig:iterative-variance}
\end{figure}

\subsection{Connections to regression/control variate methods}
\label{subsec:connections}

In this section, we show how $\hat{\mu}_{\AN}$ can be understood as a control variate/regression control method, and also how it is connected to augmented IPW (AIPW) estimation.

A popular technique for variance reduction in survey sampling is the use of regression controls, and this technique is also known in the Monte Carlo literature, where regression controls are referred to as control variates. We will now show that $\hat{\mu}_{\AN}$ is equivalent to a particular choice of regression control/control variate that is known in the Monte Carlo community, but does not appear to have been widely adopted in the survey sampling and causal inference communities.

We follow the discussion in \cite{artnotes}, where the author strongly recommends using the importance weights as control variates whenever they are known, based on the results of \cite{hesterberg}. In our setting, the random variables $w_k=I_k/p_k$ play the role of the importance weights, since they have mean 1 and re-weight the observed $Y_kI_k$ to have mean $\mu$. Then, following \cite{artnotes}, we define the family of estimators 

\begin{equation}
	\hat{\mu}_{\beta}=\frac{1}{n}\sum_{k=1}^n Y_kw_k-\beta\left( \frac{1}{n}\sum_{k=1}^n w_k-1 \right).
	\label{eq:mu-hat-beta}
\end{equation} 

Each estimator in this family is unbiased, so we choose $\beta$ to minimize variance. The variance of $\hat{\mu}_{\beta}$ is $$\frac{1}{n}\var\left( Y_kw_k \right)-\frac{2\beta}{n}\cov\left( Y_kw_k, w_k \right)+\frac{\beta^2}{n}\var(w_k)$$ which is minimized for $\beta^*=\cov(Y_k, w_k)/\var(w_k)$.

Direct computation gives 
$$\cov(Y_kw_k, w_k)=\E\left[ Y_k\frac{1-p_k}{p_k} \right]=T\quad\text{and}\quad \var(w_k)=\E\left[ \frac{1-p_k}{p_k} \right]=\pi,$$
connecting the importance sampling problem to the notation of our survey sampling problem given in \eqref{eqn:moments}. Thus, estimating the numerator and denominator of $\beta^*$ separately by $\hat{T}$ and $\hat{\pi}$ respectively gives the estimator $\hat{\beta}=\hat{T}_0/\hat{\pi}_0$ of $\beta^*$, and thus the estimator $$\hat{\mu}_{\hat{\beta}}=\frac{1}{n}\sum_{k=1}^n Y_kw_k-\frac{\hat{T}}{\hat{\pi}}\left( \frac{1}{n}\sum_{k=1}^n w_k-1 \right)=\hat{\mu}_{\text{HT}}+\frac{\hat{T}}{\hat{\pi}}\left(1-\frac{\hat{n}}{n}\right)$$ of $\mu$, which is algebraically equivalent to the estimator $\hat{\mu}_{\AN}$ in \eqref{eq:mu-hat-fp} derived via adaptive normalization. We note that the version of this estimator in \cite{artnotes, hesterberg} estimates $\cov(Y_k, w_k)$ and $\var(w_k)$ directly rather than first simplifying, a minor difference with our version.



$$\hat{\mu}_{\beta}=\beta+\frac{1}{n}\sum_{k=1}^n Y_k\frac{I_k}{p_k}-\beta\frac{I_k}{p_k}=\frac{1}{n}\sum_{k=1}^n \frac{Y_kI_k}{p_k}+\beta\left( 1-\frac{1}{n}\sum_{k=1}^n \frac{I_k}{p_k} \right).$$


\paragraph{Value of adaptive normalization as a perspective.}

While the adaptively normalized estimator we have derived is algebraically equivalent to a known estimator in the Monte Carlo literature and also a kind of AIPW estimator, we feel that our derivation, based on the simple idea of combining the denominators of $\hat{\mu}_{\text{HT}}$ and $\hat{\mu}_{\Hajek{}}$, offers a valuable and unique motivation. Furthermore, our iterative analysis of $\hat{\mu}_{\AN}$ provides instructive intuition for finite-sample variance reduction. 

Finally, despite these other guises in which the adaptively normalized estimator has appeared, it has not received significant attention in the survey sampling or causal inference community. This inattention is especially surprising when we consider that, in the Monte Carlo setting, the importance weights are often not computable in closed form, and so they cannot be used as control variates. In contrast, in the causal inference setting, the treatment probabilities are often either known from the experimental design or estimated (as propensities), and so $\hat{\mu}_{\AN}$ is usually available as an immediate improvement over $\hat{\mu}_{\text{HT}}$ or $\hat{\mu}_{\Hajek{}}$. Thus it seems that $\hat{\mu}_{\AN}$ is more widely known in the community where it is of less practical value, a gap we hope to remedy. Of course, in an applied setting, practitioners will often use other control variates or variance reduction techniques regardless, and whether or not $\hat{\mu}_{\AN}$ would be preferable to those varies from application to application. But there are still a variety of template settings in causal inference where $\hat{\mu}_{\text{HT}}$ and $\hat{\mu}_{\Hajek{}}$ are used by default, and replacing them with $\hat{\mu}_{\AN}$ can be advantageous, as discussed in the sequel.


\section{Applications}
\label{sec:apps}

In this section, we consider various causal inference settings in which the Horvitz--Thompson or \Hajek{} estimators are used ``by default,'' and show how the adaptively normalized estimator acts as a free upgrade. 

\subsection{AIPW estimation}
\label{subsec:covariate}

Our first application of adaptive normalization focuses on AIPW estimation. A basic version of AIPW was first introduced for survey sampling in the 1980s \cite{gendiff2,gendiff}, and then re-discovered and significantly expanded on for ATE estimation by the causal inference community \cite{dml,aipw}, where it is also known as the doubly-robust estimator. Our approach follows the one developed in \cite{dml}, although we concentrate on a single group for now and defer ATE estimation until further on.

We discuss AIPW estimation in the more general model where we have access to covariate information. Specifically, we assume the pairs $(Y_1, X_1),\cdots, (Y_n, X_n)$ are i.i.d.\ from a distribution $\mathcal{D}$ on $\mathbb{R}\times \mathcal{X}$, and we observe all of $X_1,\cdots, X_n$ in addition to $Y_1I_1,\cdots, Y_nI_n$ where $I_k$ are independently $\text{Ber}(p_k)$ and $p_k=p(X_k)$ is a function of the covariates. Then, as in Section~\ref{subsec:connections}, the AIPW estimate of $\mu$ is 
\begin{equation}
	\hat{\mu}_{\text{AIPW}}=\frac{1}{n}\sum_{k=1}^n \hat{\mu}(X_k)+\frac{1}{n}\sum_{k=1}^n \frac{(Y_k-\hat{\mu}(X_k))I_k}{\hat{p}(X_k)},
	\label{eq:aipw}
\end{equation}
where $\hat{\mu}(X_k)$ is an estimate of $\mu(X_k)=\E[Y_k\mid X_k]$ and $\hat{p}(X_k)$ is an estimate of $p(X_k)$. We assume that $\hat{\mu}$ and $\hat{p}$ are trained on an external training dataset $\mathcal{T}_n$ of size $n$. This is equivalent to two-fold cross-fitting with a combined sample of size $2n$; our results continue to hold for an arbitrary number of folds, but we focus on this case for ease of notation. We will require that $\hat{\mu}(\cdot)$ and $\hat{p}(\cdot)$ satisfy the following standard conditions \cite{dml}:

\begin{assumption}[Consistency]
	\label{consistency}
	As $n\to \infty$, $\hat{\mu}(\cdot)$ and $\hat{p}(\cdot)$ satisfy $$\sup_{x\in\mathcal{X}}|\mu(x)-\hat{\mu}(x)|, \sup_{x\in \mathcal{X}}|\hat{p}(x)-p(x)|\xrightarrow{\P}0.$$
\end{assumption}

\begin{assumption}[Risk decay]
	\label{risk}
	We have that $\hat{\mu}(\cdot), \hat{p}(\cdot)$ satisfy $$\E\left[ (\hat{\mu}(X_k)-\mu(X_k))^2\mid \mathcal{T}_n\right]\times\E\left[ (\hat{p}(X_k)-p(X_k))^2 \mid \mathcal{T}_n\right]=o_P(n^{-1}).$$
\end{assumption}

Our focus on $\hat{\mu}_{\text{AIPW}}$ is motivated by the fact that, given Assumptions~\ref{consistency} and \ref{risk}, a cross-fitted version of the estimator $\hat{\mu}_{\text{AIPW}}$ is semi-parametrically efficient \cite{hahn}. Of course, this efficiency result means that we cannot hope to improve $\hat{\mu}_{\text{AIPW}}$ asymptotically through adaptive normalization, but we will see that we can use our ideas to to preserve its asymptotic efficiency and reduce its finite-sample MSE in simulations.

Consider the estimator in \eqref{eq:aipw}. The first term is an estimate of $\mu$ based on imputing all of the $Y_k$ by $\hat{\mu}(X_k)$, while the second term is an inverse probability weighted estimate of the bias of the $\hat{\mu}(X_k)$. In light of our work in Section \ref{sec:adapt}, we naturally propose replacing the second term with a adaptively normalized estimator, yielding the new estimator 
\small
\begin{equation}
\muanaipw
=
\frac{1}{n}\sum_{k=1}^n \hat{\mu}(X_k)+ \frac{1}{n}\sum_{k=1}^n \frac{(Y_k-\hat{\mu}(X_k))I_k}{\hat{p}(X_k)}+\frac{1}{\hat{\pi}}\left(\sum_{k=1}^n (Y_k-\hat{\mu}(X_k))\frac{1-\hat{p}(X_k)}{\hat{p}(X_k)}\frac{I_k}{\hat{p}(X_k)}\right)\left( 1-\frac{\hat{n}}{n} \right),
\end{equation}
\normalsize
where now $$\hat{\pi}=\frac{1}{n}\sum_{k=1}^n \frac{1-\hat{p}(X_k)}{\hat{p}(X_k)}\frac{I_k}{\hat{p}(X_k)},\quad \hat{n}=\sum_{k=1}^n \frac{I_k}{\hat{p}(X_k)},$$ are functions of the estimated propensities, rather than the true treatment probabilities as in Section~\ref{sec:adapt}.

The following theorem shows that this correction is asymptotically negligible; the proof appears in Appendix~\ref{subsec:equivalence}.

\begin{theorem}
	\label{thm:aipw-equivalence}
	Suppose Assumptions \ref{bounded}-\ref{risk} hold. Then $\sqrt{n}(\hat{\mu}_{\text{AIPW}}-\muanaipw)\xrightarrow{\P} 0$.
\end{theorem}

Crucially, Theorem \ref{thm:aipw-equivalence} implies that $\muanaipw$ has the same asymptotic variance as $\hat{\mu}_{\text{AIPW}}$ under the given conditions. These are exactly the  conditions required for the efficiency of $\hat{\mu}_{\text{AIPW}}$, so we conclude that $\muanaipw$ is efficient whenever $\hat{\mu}_{\text{AIPW}}$ is. On the other hand, in finite samples, the additional correction term in $\muanaipw$ is negatively correlated with the other terms, and reduces variance, a fact we demonstrate empirically in Section \ref{sec:sims}. Taken together, these observations suggest that $\muanaipw$ should typically be preferred to $\hat{\mu}_{\text{AIPW}}$.

Our proposal is by no means the only attempt to improve on $\hat{\mu}_{\text{AIPW}}$. However, our proposal differs from existing work in two important ways. First, many other proposals are meant to improve on AIPW in the case when $\hat{\mu}$ is misspecified, which is not our motivation here, although we do explore this setting in simulation and see that $\muanaipw$ handles misspecification better than $\hat{\mu}_{\text{AIPW}}$. Second, and more importantly, other proposed estimators are significantly more complex, and sometimes requires solving certain estimating equations numerically. In contrast, $\muanaipw$ has an explicit, simple, closed form, and computing it requires nothing more than what is required to compute $\hat{\mu}_{\text{AIPW}}$.

\subsection{ATE estimation}
\label{subsec:ate}

Another setting in which IPW estimators play a key role is in the problem of average treatment effect (ATE) estimation. Our model for ATE estimation is that there are triplets $$(Y_1(1), Y_1(0), p_1),\ldots, (Y_n(1), Y_n(0), p_n)$$ drawn i.i.d.\ from a distribution $\mathcal{D}$ on $\R\times \R\times [0,1]$ satisfying Assumption \ref{bounded} for both $Y_k(1)$ and $Y_k(0)$. We assume that we observe $Y_1(I_1),\ldots, Y_n(I_n)$ where $I_k\mid p_k\sim \text{Ber}(p_k)$, and want to estimate the ATE $\tau=\E[Y_k(1)-Y_k(0)]$ from these observations. There are a variety of general approaches to estimating $\tau$ \cite{imbens}; the one relevant to our work is the approach of using IPW estimators to separately estimate $\mu_1=\E[Y_k(1)]$ and $\mu_0=\E[Y_k(0)]$, and then subtracting the two estimates. 

Since the problems of estimating $\mu_1$ and $\mu_0$ are survey sampling problems, they are often estimated with Horvitz--Thompson or \Hajek{} estimators. Continuing our general theme, why not use $\hat{\mu}_{\AN}$ instead? Thus, we propose 
\begin{equation}
\hat{\tau}_{\AN}=\hat{\mu}_{1,\AN}-\hat{\mu}_{0,\AN},
\label{eq:ate-an}
\end{equation}
where $\hat{\mu}_{1,\AN}$ is the adaptively normalized estimator based on $Y_1(1)I_1,\ldots, Y_n(1)I_n$ and $\hat{\mu}_{0, \AN}$ is the same based on $Y_1(0)(1-I_1),\ldots, Y_n(0)(1-I_n)$. 

This is not the only possible use of adaptive normalization in this setting. Note that $\hat{\mu}_{1, \AN}$ and $\hat{\mu}_{0, \AN}$ here are designed to minimize variance within each group, even though our target estimand is the difference between groups. It is natural to ask what happens if we instead attempt to directly minimize the variance of the estimated difference between the two groups; we discuss this approach in detail and compare it to $\hat{\tau}_{\AN}$ in Appendix~\ref{sec:joint}. We conclude that estimating the two groups separately and taking the difference is reasonable and generally advised. 

Similarly, rather than using AIPW estimation in each group for ATE estimation, we can also use adaptively normalized AIPW estimation in each group instead. Indeed, we will show in Section~\ref{sec:sims} that these estimators improve on both of the usual estimators in simulation.

\subsection{Policy evaluation}
\label{subsec:policy}

The final setting in which we propose replacing an IPW estimator with an adaptively normalized estimator is policy evaluation. In this context, we wish to learn a statistical rule for assigning treatments to a population that maximizes the total welfare. We follow the regret minimization framework introduced by \cite{manski} and build directly on the work of \cite{kitagawa}, although much of our notation is drawn from \cite{wager}.

Formally, suppose that individual $k$ has potential outcomes $Y_k(1)$ and $Y_k(0)$ depending on whether or not they receive a treatment, and we wish to learn a policy $\pi$ that maps known covariates $X_k\in \mathcal{X}$ to a treatment assignment in $\{0,1\}$. The value of a policy $\pi$ (note that we are no longer using $\pi$ to represent moments of the unknown distribution as in Section~\ref{sec:adapt}) is $V(\pi)=\E[Y_i(\pi(X_i))]$, the average outcome of an individual treated using the policy. Assuming we are restricted to a class of policies $\Pi$, the best possible policy is $\pi^*=\arg\max_{\pi'\in \Pi} V(\pi')$, and we evaluate a policy based on its regret $V(\pi^*)-V(\pi)$. 

Ideally we would learn a policy by maximizing $V$, but we cannot compute $V$. Instead we estimate $V$ from historical data $(Y_1(I_1), X_1),\cdots, (Y_n(I_n), X_n)$ where $I_k\sim \text{Ber}(p(X_k))$ is an indicator of whether or not $Y_k$ received the treatment, and we assume that Assumption~\ref{bounded} holds. Under the assumption that the propensity map $p(\cdot)$ is known, \cite{kitagawa} proposed $$\hat{V}_{\text{IPW}}(\pi)=\frac{1}{n}\sum_{k=1}^n \frac{\mathbf{1}\left\{ I_k=\pi(X_k) \right\}Y_k}{\P(I_k=\pi(X_k)\mid X_k)},$$ as an unbiased estimate of $V(\pi)$.

Of course, at this point we can sense that a better estimate of $V$ would be the adaptively normalized 
\begin{equation}
	\hat{V}_{\text{AN}}=\hat{V}_{\text{IPW}}(\pi)+\frac{\sum_{k=1}^n Y_k\frac{1-\P(I_k=\pi(X_k)\mid X_k)}{\P(I_k=\pi(X_k)\mid X_k)}\frac{\mathbf{1}\left\{ I_k=\pi(X_k) \right\}}{\P(I_k=\pi(X_k)\mid X_k)}}{\sum_{k=1}^n\frac{1-\P(I_k=\pi(X_k)\mid X_k)}{\P(I_k=\pi(X_k)\mid X_k)}\frac{\mathbf{1}\left\{ I_k=\pi(X_k) \right\}}{\P(I_k=\pi(X_k)\mid X_k)}}\left( 1-\frac{1}{n}\sum_{k=1}^n \frac{\mathbf{1}\left\{ I_k=\pi(X_k) \right\}}{\P(I_k=\pi(X_k)\mid X_k)}\right).
	\label{eq:v-hat}
\end{equation}

The main result of \cite{kitagawa} showed that maximizing $\hat{V}_{\text{IPW}}(\pi)$ yields a policy whose regret decays at a rate of $1/\sqrt{n}$. We now give an analogous result for $\hat{V}_{\text{AN}}$.
\begin{theorem}
\label{thm:regret-bound}
	Let $\hat{\pi}_{\text{AN}}=\argmin_{\pi\in \Pi} \hat{V}_{\text{AN}}(\pi)$ and suppose that $\Pi$ has finite VC-dimension. Then $$\E\left[ V(\pi^*)-V(\hat{\pi}_{\text{AN}}) \right]\leq O\left( \frac{M}{\delta}\sqrt{\frac{\text{VC}(\Pi)}{n}} \right).$$
\end{theorem}
Thus, $\hat{V}_{\AN}$ preserves the theoretical guarantees associated with $\hat{V}_{\text{IPW}}$, and we verify in the next section that policies learned with $\hat{V}_{\text{AN}}$ are closer to the optimal policy. One drawback of $\hat{V}_{\text{AN}}$ however, is that $\hat{V}_{\text{IPW}}$ can be interpreted as a weighted classification objective \cite{wager, kitagawa}, facilitating optimization, while $\hat{V}_{\text{AN}}$ unfortunately does not have such an interpretation. Finally, we note that \cite{wager} introduced the idea of policy learning based on AIPW estimation instead of IPW estimation; we consider the possibility of extending our ideas from Section~\ref{subsec:covariate} in this direction an exciting opportunity for future work.


\section{Experiments}
\label{sec:sims}

In this section, we use a series of experiments to evaluate the empirical performance of adaptive normalization. The goal of our experiments is to validate that applying adaptive normalization for mean estimation as well as in the various settings of Section~\ref{sec:apps} actually pays the dividends we expect. 

The first portion of our experiments are focused on survey sampling as discussed in Section~\ref{sec:adapt} and ATE estimation as discussed in Section~\ref{subsec:ate}. Each of these serves as an opportunity to compare $\hat{\mu}_{\text{AN}}$ to $\hat{\mu}_{\text{HT}}$ and $\hat{\mu}_{\Hajek{}}$, and also an opportunity to compare $\muanaipw$ of Section~\ref{subsec:covariate} to $\hat{\mu}_{\text{AIPW}}$, since the AIPW estimator can be used for both survey sampling and ATE estimation. The second portion of our experiments focus on policy estimation, as introduced in Section~\ref{subsec:policy}, where we show that minimizing $\hat{V}_{\text{AN}}$ learns better policies than minimizing $\hat{V}_{\text{IPW}}$. Additional simulations, including real-data experiments, appear in Appendix~\ref{sec:sims_app}.

\subsection{Survey and ATE experiments}
\label{subsec:models}

\paragraph{Data generating models.}
The goal of our simulations is to compare the various proposed estimators to each other and existing estimators in different settings. To this end, we consider two models, one of which represents a fairly benign setting in which $Y_k$ and $p_k$ are approximately jointly normal, and estimating $\mu$ is not too difficult. The second model represents a more pathological setting, where $Y_k$ and $p_k$ have an extremely strong negative relationship, and so estimating $\mu$ becomes significantly more challenging. We now describe these models in detail.

The first model, which we refer to as our \emph{normal model}, is a model in which we generate
\begin{equation}
	(Y_k, \tilde{p}_k)\sim N\left( \left[ \begin{array}{c} \mu\\ 0\end{array} \right], \left[ \begin{array}{cc} 1&\theta\\ \theta&1 \end{array} \right] \right),\quad p_k=\Phi^{-1}(\tilde{p}_k),\label{eq:normal-model}
\end{equation}
where $\Phi$ is the normal CDF. This model produces pairs $(Y_k, p_k)$ for which $p_k$ is marginally uniform on $[0,1]$ and the correlation between $Y_k$ and $p_k$ is approximately $\theta$. The correlation is not exactly $\theta$, owing to the non-linearity of $\Phi^{-1}$, but the parameter $\theta$ still has the desired effect of making $Y_k$ and $p_k$ more or less correlated. Additionally, we truncate both $p_k$ and $Y_k$ so that they satisfy Assumption \ref{bounded} with $M=50$ and $\delta=0.01$.

The second model we consider is our \emph{power law model}, where we generate 
\begin{equation}
p_k\sim\text{Uni}(\epsilon, 1),\quad Y_k=p_k^{-\alpha}+Z_k,\quad Z_k\sim N(0, \sigma^2).
\label{eq:powerlaw-model}
\end{equation}
This model captures settings where $Y$ and $p$ have a non-linear negative relationship. When $\alpha=0$, $Y$ and $p$ are independent, and as $\alpha$ increases, the strength of the negative relationship increases. The parameter $\sigma^2$ is an arbitrary noise parameter which we set to $\sigma^2=9$. As before, we truncate the $p_k$ and $Y_k$ to satisfy Assumption \ref{bounded} with $M=10^6$ and $\delta=10^{-3}$.

For problems with covariates, we incorporate covariate information $X_k$ into either model by taking $X_k=p_k$, so that the propensity mapping is the identity. For AIPW estimation, we assume the propensity map is known and estimated exactly, while $\E[Y_k\mid X_k]$ is estimated with linear regression and two-fold cross-fitting. This is nearly well-specified in the normal model, where where $Y_k$ and $p_k$ are approximately multivariate normal, but is severely misspecified in the power law model. For ATE estimation problems, we take $Y_k(1)=Y_k$ and $Y_k(0)=Y_k-\tau$ for a constant treatment effect of $\tau=0.5$.

\begin{figure}[t]
	\centering
	\begin{tabular}{cc}\includegraphics[width=0.48\linewidth]{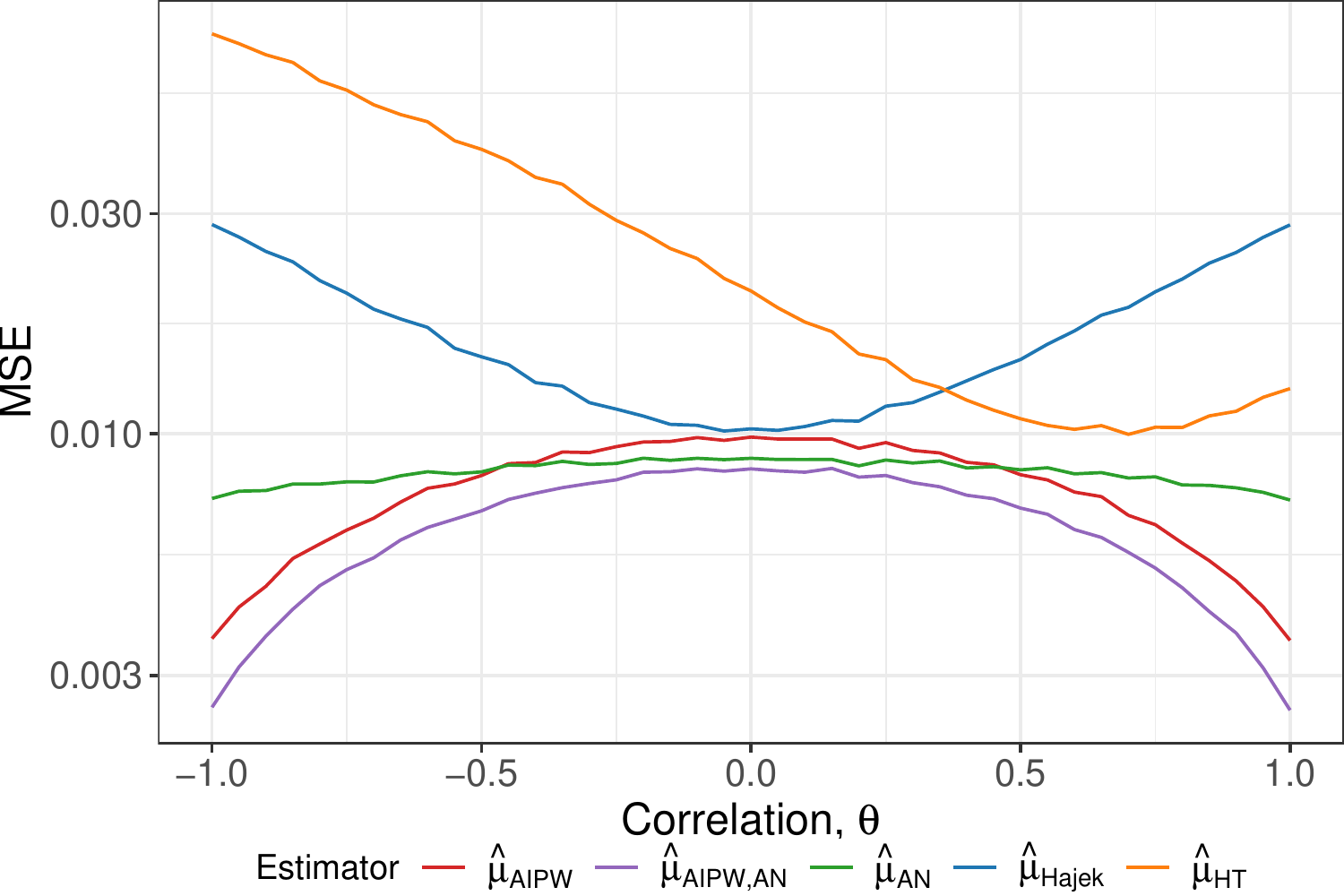}&\includegraphics[width=0.48\linewidth]{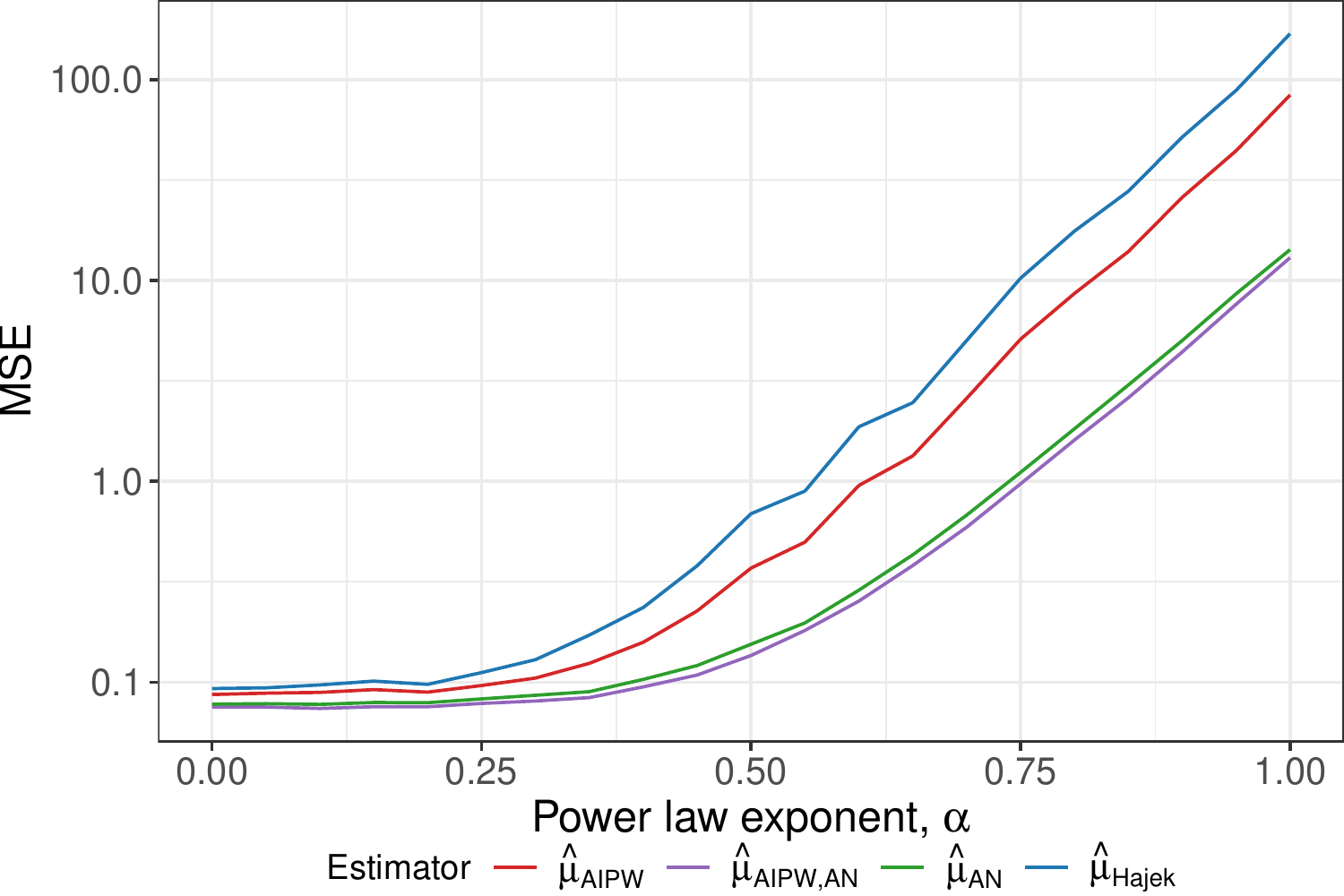}\end{tabular}
	\caption{Left: the estimated MSE of all discussed estimators on data generated from the normal model \eqref{eq:normal-model} with $n=500$ and $\mu=1$ for different values of $\theta$. Right: the estimated MSE of the same estimators on data generated from the power law model \eqref{eq:powerlaw-model} with $n=500$ and $\alpha$ varying (the value of $\mu$ depends on $\alpha$). All MSEs are averaged across 20,000 trials. Across both models and the full range of parameter values, we see that the adaptively normalized estimators improve on the traditional estimators.}
	\label{fig:mse-compare}
\end{figure}

\paragraph{Survey sampling simulations.}
We simulate our data from both the normal model and the power law model, fixing $n=500$ but allowing the parameters $\theta$ and $\alpha$ to vary so as to explore a range of different possible relationships between $Y$ and $p$. The results are shown in Figure \ref{fig:mse-compare}, omitting $\hat{\mu}_{\text{HT}}$ in the power law model to preserve the scale of the plot. We split our discussion into two portions, focusing separately on the estimators that do and do not use covariate information.

Among the estimators that do not use covariate information, namely $\hat{\mu}_{\text{HT}}$, $\hat{\mu}_{\Hajek{}}$, and $\hat{\mu}_{\AN}$, it is clear that $\hat{\mu}_{\AN}$ is the best choice. In the normal model, the MSE of both $\hat{\mu}_{\text{HT}}$ and $\hat{\mu}_{\Hajek{}}$ is significantly affected by the strength of the correlation, as controlled by the parameter $\theta$, while $\hat{\mu}_{\AN}$ is able to do well regardless of the correlation structure in the problem, and always has the lowest MSE in this group. In the power law model, which represents a more challenging problem, the difference is even more stark: the MSE of $\hat{\mu}_{\AN}$ is, for larger values of $\alpha$, an order of magnitude lower than that of the other two estimators.

Among the estimators that do use covariate information, namely $\hat{\mu}_{\text{AIPW}}$ and $\muanaipw$, there a small but noticeable difference in the normal model. However, in the power law model, where the estimates of $\E[Y_k\mid X_k]$ are misspecified, there is a more substantial difference. Because of the misspecification, $\hat{\mu}_{\text{AIPW}}$ offers only the slightest of improvements over $\hat{\mu}_{\Hajek{}}$. In contrast, $\muanaipw$ is noticeably better than either of $\hat{\mu}_{\text{AIPW}}$ or $\hat{\mu}_{\Hajek{}}$, but is essentially identical to $\hat{\mu}_{\AN}$ (again, the lack of improvement from $\hat{\mu}_{\AN}$ to $\muanaipw$ is caused by model misspecification). This suggests that the additional care taken in $\hat{\mu}_{\text{AIPW}, \AN}$ to estimate the bias of $\hat{\mu}(\cdot)$ is especially important when $\hat{\mu}(\cdot)$ is misspecified and this bias is large.

\paragraph{ATE simulations.}

\begin{figure}[t]
	\centering
	\begin{tabular}{cc}\includegraphics[width=0.48\linewidth]{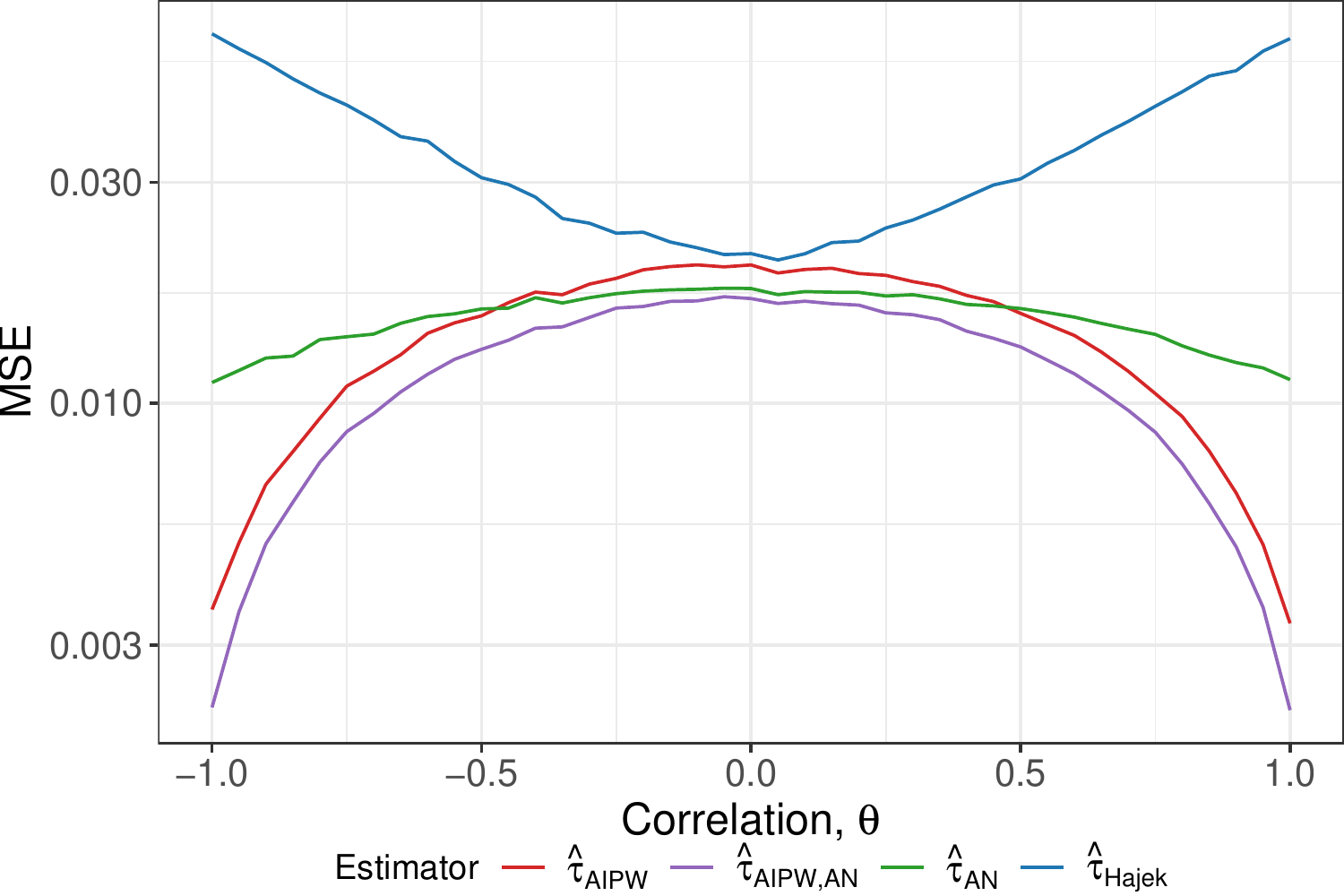}&\includegraphics[width=0.48\linewidth]{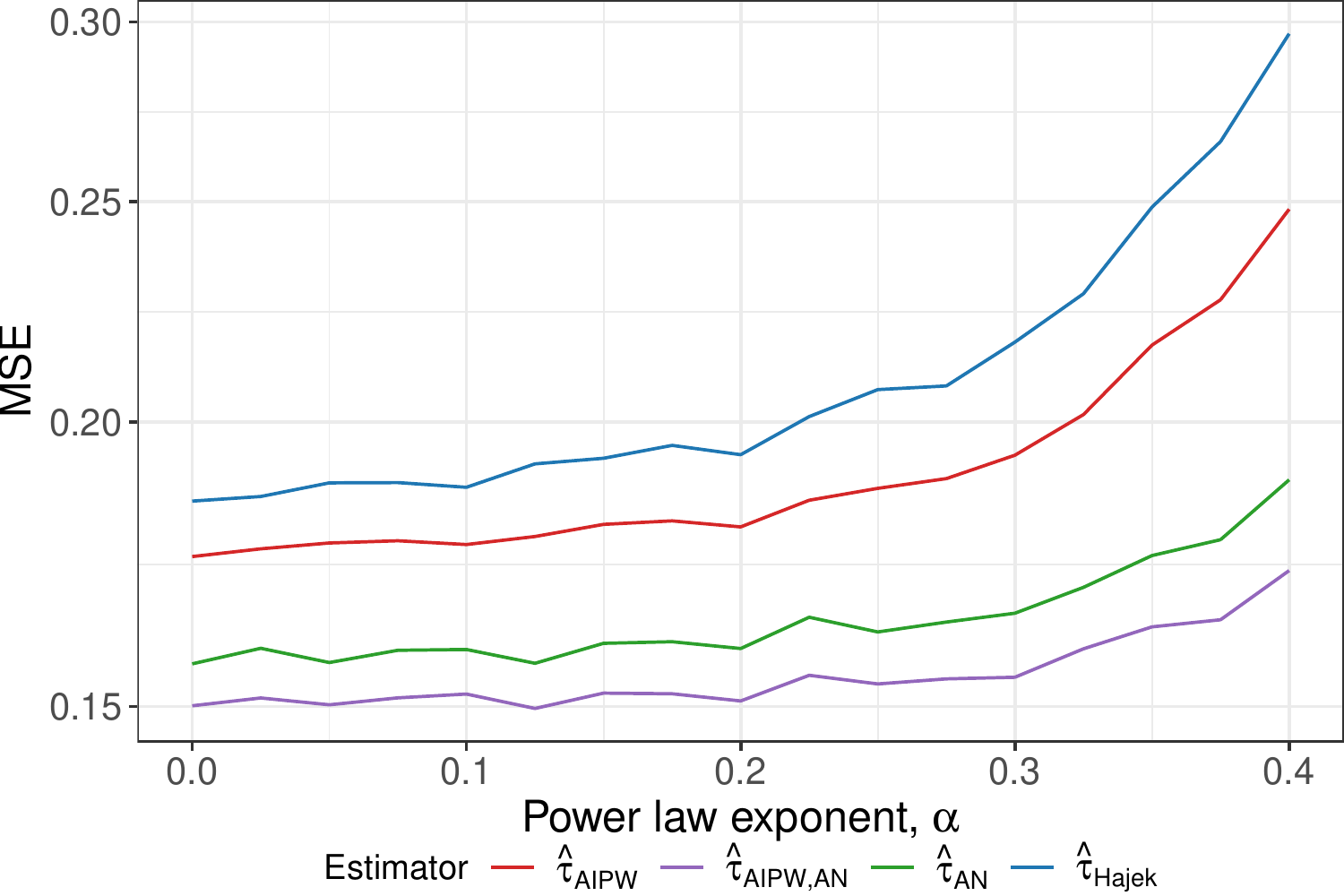}\end{tabular}
	\caption{Left: estimated MSE for estimators applied to data from the normal model \eqref{eq:normal-model} for ATE estimation with $n=500$, $\mu=1$, and $\tau=0.5$ for different values of $\theta$. Right: the estimated MSE of the same estimators applied to data from the power law model \eqref{eq:powerlaw-model} for ATE estimation with $n=500$ and $\tau=0.5$ for different value of $\alpha$. All MSEs are averaged over 20,000 trials. In all cases, adaptive normalization improves MSE.} 
	\label{fig:ate-compare}
\end{figure}

As above, we simulate in both models, fixing $n=500$ and $\tau=0.5$ while varying $\theta$ and $\alpha$. Here, we omit $\hat{\tau}_{\text{HT}}$ from both plots, since its MSE is so large as to distort the scale of the plot. Our main observations from the survey sampling setting largely carry over: $\hat{\tau}_{\AN}$ is consistently better than $\hat{\tau}_{\Hajek{}}$, and $\hat{\tau}_{\text{AIPW}, \AN}$ is slightly better than $\hat{\tau}_{\text{AIPW}}$ when the model for $Y_k$ is well-specified and significantly better when the model for $Y_k$ is misspecified. 

\subsection{Policy evaluation experiments}

Our second set of simulations compares the two policy learning objectives, $\hat{V}_{\text{IPW}}$ and $\hat{V}_{\text{AN}}$ of Section~\ref{subsec:policy}. 

\paragraph{Data generating model.}
Due to the slightly different set-up of the problem, we deviate from the models introduced above and instead use a model similar to the one used in the simulation studies of \cite{wager}. For each $k$, we take 
\begin{equation}
	X_k\sim N(0, I_{3\times 3}),\quad p(X_k)=\frac{1}{1+\exp(-X_{k,1})},\quad Y_k(1)=X_{k,1}, Y_k(0)=Y_k(1)-\text{sgn}(X_{k,2}+X_{k,3}),
	\label{eq:policy-model}
\end{equation}
where $X_{k,i}$ is the $i^{\text{th}}$ entry of $X_k$.

In general the problem of minimizing $\hat{V}_{\text{IPW}}$ or $\hat{V}_{\text{AN}}$ is non-convex. To obtain a tractable problem, we restrict the class $\Pi$ to be the set of all policies of the form $\pi(X_k)=\mathbf{1}\{X_{k,2}>T\}$ for $T\in[-1,1]$. Note that this is a misspecified setting, in the sense that the optimal policy $\pi(X_k)=\mathbf{1}\{X_{k,2}+X_{k,3}>0\}$ is not in the class we are optimizing over.

\paragraph{Policy learning simulations.}

For each of $\hat{V}_{\text{IPW}}$ and $\hat{V}_{\text{AN}}$, we generate a sample of size $n$ from \eqref{eq:policy-model} and learn a cut-off $T$ that minimizes the corresponding policy objective by grid search on $T\in[-1,1]$. The average threshold learned, for a range of values of $n$, is shown in Table~\ref{tab:policy}.

The optimal policy is to threshold at $T=0$, and so we take this as a point of comparison. We see that the thresholds learned from minimizing $\hat{V}_{\text{AN}}$ are consistently closer to the optimal threshold of zero than those learned by minimizing $\hat{V}_{\text{IPW}}$, and that the gap between the performance of the two objectives is consistent across the range of values of $n$ we consider. We note that, although the improvements we obtain are small in absolute terms, this may be a consequence of the simplicity of the policy class we consider. In particular, \cite{wager} finds little improvement from using AIPW estimators for policy learning of depth-1 trees, but more substantial improvements for policy learning of depth-2 trees, suggesting that $\hat{V}_{\text{AN}}$ may also be a more substantial improvement over $\hat{V}_{\text{IPW}}$ when learning policies from more complex classes. 

\begin{table}[t]
	\centering
	\begin{tabular}{lllll}
		\toprule
		&\multicolumn{4}{c}{\emph{Sample size}}\\
		\cmidrule{2-5}
		\emph{Objective}&$n=250$&$n=500$&$n=750$&$n=1000$\\
		\midrule
		$\hat{V}_{\text{IPW}}$&$-0.057\pm 0.0013$&$-0.035\pm 0.0011$&$-0.026\pm 0.0011$&$-0.020\pm 0.0014$\\
		$\hat{V}_{\text{AN}}$&$-0.039\pm 0.0018$&$-0.015\pm 0.0014$&$-0.010\pm 0.0009$&$-0.004\pm 0.0009$\\
		\bottomrule
	\end{tabular}
	\caption{Thresholds learned by optimizing $\hat{V}_{\text{IPW}}$ and $\hat{V}_{\text{AN}}$ on samples of different sizes of data generated according to \eqref{eq:policy-model}. Each entry is the average threshold chosen over 100,000 trials and standard errors are over 10 replications. The optimal policy is to threshold at 0, so we see that minimizing $\hat{V}_{\text{AN}}$ consistently learns better thresholds.}
	\label{tab:policy}
\end{table}

\section{Discussion}

In this paper we study \emph{adaptive normalization} for IPW estimators: rather than normalizing by the sample size or by the sum of the weights, we propose normalizing by a data-dependent affine combination of the two. For mean estimation in survey sampling, our proposed estimator is algebraically equivalent to a control variate method from the Monte Carlo literature that has not been studied in the causal inference literature before. We further develop the adaptive normalization idea in causal inference settings by using it to improve on the AIPW estimators, to propose new estimators for the ATE in randomized experiments, and new objectives for policy learning.

There are several possible future directions for this work. One is to relax the assumption that the treatment indicators $I_k$ are independent, which is unrealistic in many observational datasets and directly violated in certain experimental designs. However, if the correlation structure of the $I_k$ is known or estimated, analogues of our limit theorems and estimators could be developed. In a different direction, there are many places within and beyond causal inference where inverse probability weighted estimators are used in context-specific ways, such as off-policy evaluation on networks \cite{seeding}, recommender system evaluation \cite{schnabel2016recommendations}, in learning to rank problems \cite{oosterhuis2021unifying}, and inference from bandits \cite{hadad2019confidence,bibaut2021post}. Developing and applying similar ideas in these contexts suggests many promising lines of future work.

\subsection*{Acknowledgements}
We thank Guillaume Basse, Alex Chin, Dean Eckles, Kevin Guo, Ramesh Johari, Brian Karrer, and Fredrik S\"avje for helpful discussions and feedback on early versions of this work. This work was partially supported by ARO award 73348-NS-YIP.
%
%

\bibliographystyle{plain}

\bibliography{TukeyTrotterDraft}

\begin{thebibliography}{10}

\bibitem{wager}
Susan Athey and Stefan Wager.
\newblock Policy learning with observational data.
\newblock {\em Econometrica}, 89(1):133--161, 2021.

\bibitem{hajek}
D~Basu.
\newblock An essay on the logical foundations of survey sampling, part i.
\newblock In VP~Godambe and DA~Sprott, editors, {\em Foundations of Statistical
  Inferences}. Holt, Rinehart and Winston, Toronto, Canada, 1971.

\bibitem{bibaut2021post}
Aur{\'e}lien Bibaut, Antoine Chambaz, Maria Dimakopoulou, Nathan Kallus, and
  Mark van~der Laan.
\newblock Post-contextual-bandit inference.
\newblock {\em arXiv preprint arXiv:2106.00418}, 2021.

\bibitem{gendiff2}
Claes~M. Cassel, Carl~E. Sarndal, and Jan~H. Wretman.
\newblock Some results on generalized difference estimation and generalized
  regression estimation for finite populations.
\newblock {\em Biometrika}, 63(3):615--620, 1976.

\bibitem{el-aipw}
Song~Xi Chen, Denis H.~Y. Leung, and Jing Qin.
\newblock Improving semiparametric estimation by using surrogate data.
\newblock {\em Journal of the Royal Statistical Society: Series B (Statistical
  Methodology)}, 70(4):803--823, 2008.

\bibitem{dml}
Victor Chernozhukov, Denis Chetverikov, Mert Demirer, Esther Duflo, Christian
  Hansen, Whitney Newey, and James Robins.
\newblock {Double/debiased machine learning for treatment and structural
  parameters}.
\newblock {\em The Econometrics Journal}, 21(1):C1--C68, 01 2018.

\bibitem{seeding}
Alex Chin, Dean Eckles, and Johan Ugander.
\newblock Evaluating stochastic seeding strategies in networks.
\newblock {\em Management Science}, 2021.

\bibitem{savje-consistency}
Ang{\`e}le Delevoye and Fredrik S{\"a}vje.
\newblock Consistency of the horvitz--thompson estimator under general sampling
  and experimental designs.
\newblock {\em Journal of Statistical Planning and Inference}, 207:190--197,
  2020.

\bibitem{firth}
David Firth.
\newblock {On improved estimation for importance sampling}.
\newblock {\em Brazilian Journal of Probability and Statistics}, 25(3):437 --
  443, 2011.

\bibitem{admissibility}
V.~P. Godambe and V.~M. Joshi.
\newblock Admissibility and bayes estimation in sampling finite populations. i.
\newblock {\em The Annals of Mathematical Statistics}, 36(6):1707--1722, 1965.

\bibitem{hadad2019confidence}
Vitor Hadad, David~A Hirshberg, Ruohan Zhan, Stefan Wager, and Susan Athey.
\newblock Confidence intervals for policy evaluation in adaptive experiments.
\newblock {\em arXiv preprint arXiv:1911.02768}, 2019.

\bibitem{hahn}
Jinyong Hahn.
\newblock On the role of the propensity score in efficient semiparametric
  estimation of average treatment effects.
\newblock {\em Econometrica}, 66(2):315--331, 1998.

\bibitem{is}
J.~M. Hammersley and D.~C. Handscomb.
\newblock {\em Monte Carlo Methods}.
\newblock Springer, 1964.

\bibitem{gmm}
Bruce~E Hansen and Seojeong Lee.
\newblock Inference for iterated gmm under misspecification.
\newblock {\em Econometrica}, 89(3):1419--1447, 2021.

\bibitem{hesterberg}
Tim Hesterberg.
\newblock Weighted average importance sampling and defensive mixture
  distributions.
\newblock {\em Technometrics}, 37(2):185--194, 1995.

\bibitem{horvitz-thompson}
D.~G. Horvitz and D.~J. Thompson.
\newblock A generalization of sampling without replacement from a finite
  universe.
\newblock {\em Journal of the American Statistical Association},
  47(260):663--685, 1952.

\bibitem{imbens}
Guido~W Imbens.
\newblock Nonparametric estimation of average treatment effects under
  exogeneity: A review.
\newblock {\em Review of Economics and Statistics}, 86(1):4--29, 2004.

\bibitem{kitagawa}
Toru Kitagawa and Aleksey Tetenov.
\newblock Who should be treated? empirical welfare maximization methods for
  treatment choice.
\newblock {\em Econometrica}, 86(2):591--616, 2018.

\bibitem{gendiff}
Roderick J.~A. Little.
\newblock Estimating a finite population mean from unequal probability samples.
\newblock {\em Journal of the American Statistical Association},
  78(383):596--604, 1983.

\bibitem{manski}
Charles~F Manski.
\newblock Statistical treatment rules for heterogeneous populations.
\newblock {\em Econometrica}, 72(4):1221--1246, 2004.

\bibitem{oosterhuis2021unifying}
Harrie Oosterhuis and Maarten de~Rijke.
\newblock Unifying online and counterfactual learning to rank: A novel
  counterfactual estimator that effectively utilizes online interventions.
\newblock In {\em Proceedings of the 14th ACM International Conference on Web
  Search and Data Mining}, pages 463--471, 2021.

\bibitem{artnotes}
Art~B. Owen.
\newblock {\em Monte Carlo theory, methods and examples}.
\newblock 2013.

\bibitem{el-aipw-2}
Jing Qin, Biao Zhang, and Denis H.~Y. Leung.
\newblock Empirical likelihood in missing data problems.
\newblock {\em Journal of the American Statistical Association},
  104(488):1492--1503, 2009.

\bibitem{aipw}
James~M. Robins, Andrea Rotnitzky, and Lue~Ping Zhao.
\newblock Estimation of regression coefficients when some regressors are not
  always observed.
\newblock {\em Journal of the American Statistical Association},
  89(427):846--866, 1994.

\bibitem{robinson-consistency}
P.~M. Robinson.
\newblock On the convergence of the horvitz-thompson estimator.
\newblock {\em Australian Journal of Statistics}, 24(2):234--238, 1982.

\bibitem{improved-aipw}
Andrea Rotnitzky, Quanhong Lei, Mariela Sued, and James~M. Robins.
\newblock {Improved double-robust estimation in missing data and causal
  inference models}.
\newblock {\em Biometrika}, 99(2):439--456, 04 2012.

\bibitem{sarndal}
Carl-Erik S{\"a}rndal, Bengt Swensson, and Jan Wretman.
\newblock {\em Model assisted survey sampling}.
\newblock Springer Science \& Business Media, 2003.

\bibitem{schnabel2016recommendations}
Tobias Schnabel, Adith Swaminathan, Ashudeep Singh, Navin Chandak, and Thorsten
  Joachims.
\newblock Recommendations as treatments: Debiasing learning and evaluation.
\newblock In {\em international conference on machine learning}, pages
  1670--1679. PMLR, 2016.

\bibitem{bandit}
Adith Swaminathan and Thorsten Joachims.
\newblock The self-normalized estimator for counterfactual learning.
\newblock In {\em Advances in Neural Information Processing Systems}, pages
  3231--3239. Citeseer, 2015.

\bibitem{sampling}
Yves Till{\'e} and Alina Matei.
\newblock {\em sampling: Survey Sampling}, 2021.
\newblock R package version 2.9.

\bibitem{trottertukey}
Hale~F. Trotter and John~W. Tukey.
\newblock Conditional monte carlo for normal samples.
\newblock In {\em Symposium on Monte Carlo Methods}, 1954.

\end{thebibliography}

\appendix


\section{Technical proofs}
\label{sec:proofs}

\subsection{Proofs of Theorems \ref{thm:mu-hat-lambda-clt} and \ref{thm:mu-hat-fp-clt}} 

In this subsection, we detail the calculations underlying the CLTs in Section \ref{sec:adapt}. The building block of our results is a joint CLT for the vector $\hat{\beta}=(\hat{S}/n, \hat{T}, \hat{\pi}, \hat{n}/{n})$, which has mean $\beta=(\mu, T, \pi, 1)$.

\begin{lemma}
	\label{lem:clt}
	Under Assumption \ref{bounded}, we have $$\sqrt{n}(\hat{\beta}-\beta)\xrightarrow{d} N(0, \Sigma),$$ where the entries of $\Sigma$ are given by

	\begin{multicols}{2}
		$$n\var(\hat{S}/n)=\E\left[ Y_k^2\frac{1-p_k}{p_k} \right]$$
$$n\var(\hat{T})=\E\left[ Y_k^2\frac{(1-p_k)^3}{p_k^3} \right]$$
$$n\var(\hat{\pi})=\E\left[ \frac{(1-p_k)^3}{p_k^3} \right]$$
$$n\var(\hat{n}/n)=\E\left[ \frac{1-p_k}{p_k} \right]$$
$$n\cov\left( \hat{S}/n, \hat{T} \right)=\E\left[ Y_k^2\frac{(1-p_k)^2}{p_k^2} \right]$$
$$n\cov\left( \hat{S}/n, \hat{\pi} \right)=\E\left[ Y_k\frac{(1-p_k)^2}{p_k^2} \right]$$
$$n\cov\left( \hat{S}/n, \hat{n}/n \right)=\E\left[ Y_k\frac{1-p_k}{p_k} \right]$$
$$n\cov\left( \hat{T}, \hat{\pi} \right)=\E\left[ Y_k\frac{(1-p_k)^3}{p_k^3} \right]$$
$$n\cov\left( \hat{T}, \hat{n}/n \right)=\E\left[ Y_k\frac{(1-p_k)^2}{p_k^2} \right]$$
$$n\cov\left( \hat{\pi}, \hat{n}/n \right)=\E\left[ \frac{(1-p_k)^2}{p_k^2} \right].$$
\end{multicols}
\end{lemma}

\begin{proof}
	For any vector $v\in \R^4$, the quantity $\sqrt{n}(v^T\hat{\beta}-v^T\beta)$ is a sum of i.i.d.\ random variables. These random variables have finite variance by Assumption \ref{bounded}, so the classical Lindeberg CLT applies with a limiting distribution that is $N(0, v^T\Sigma v)$. The lemma then follows from the Cramer-Wold device.
\end{proof}

The CLTs for $\hat{\mu}_{\lambda}$ and $\hat{\mu}_{\text{AN}}$ now follow from applications of the delta method to different functions of $\beta$. For $\hat{\mu}_{\lambda}$, we will need only the first and last coordinate of $\beta$, since $\hat{\mu}_{\lambda}$ depends only on these, while $\hat{\mu}_{\text{AN}}$ is a function of all four coordinates of $\beta$.

\begin{proof}[Proof of Theorem \ref{thm:mu-hat-lambda-clt}]
	Note that $\hat{\mu}_{\lambda}=f(\hat{\beta})$ for $f(x,y,z,w; \lambda)=\frac{x}{\lambda+(1-\lambda)w}$ The function $f$ is differentiable at the point $\beta$, and has gradient $\nabla_{\beta}f=(1, 0, 0, -(1-\lambda)\mu)$, and so by the delta method we conclude that $\hat{\mu}_{\lambda}$ satisfies a CLT with the asymptotic variance $$\nabla_{\beta}f^T\Sigma\nabla_{\beta}f=\E\left[ Y_k^2\frac{1-p_k}{p_k} \right]-2(1-\lambda)\mu\E\left[ Y_k\frac{1-p_k}{p_k} \right]+(1-\lambda)^2\mu^2\E\left[ \frac{1-p_k}{p_k} \right],$$ which rearranges to the variance in (\ref{eq:mu-hat-lambda-clt}), as desired. 
\end{proof}

\begin{proof}[Proof of Theorem \ref{thm:mu-hat-fp-clt}]
We have $\hat{\mu}_{\text{AN}}=f(\hat{\beta})$ for $f(x,y,z,w)=x+\frac{y}{z}(1-w)$. Taking gradients gives $\nabla_{\beta}f=(1, 0, 0, -T/\pi)$, and so the asymptotic variance of $\hat{\mu}_{\text{AN}}$ is $$\nabla_{\beta}f^T\Sigma\nabla_{\beta}f=\E\left[ Y_k^2\frac{1-p_k}{p_k} \right]-\frac{2T}{\pi}\E\left[ Y_k\frac{1-p_k}{p_k} \right]+\frac{T^2}{\pi^2}\E\left[ \frac{1-p_k}{p_k} \right],$$ which simplifies to the variance in (\ref{eq:mu-hat-fp-clt}), as desired.
\end{proof}

\subsection{Proof of Theorem \ref{thm:denom-convergence}}
\label{app:subsec:proof2}

This subsection contains the proof of Theorem \ref{thm:denom-convergence} and details for the heuristic argument for variance reduction given in the main text. 

As before, we combine the two equations in (\ref{eq:updates}) to obtain $$\hat{\mu}^{(t)}=\frac{\hat{S}}{\left( 1-\frac{\hat{T}}{\hat{\pi}\hat{\mu}^{(t-1)}} \right)n+\frac{\hat{T}}{\hat{\pi}\hat{\mu}^{(t-1)}}\hat{n}}=\frac{\hat{S}/n}{1-\frac{\hat{T}}{\hat{\pi}\hat{\mu}^{(t-1)}}\left( 1-\frac{\hat{n}}{n}\right)},$$ and write this as 
$\hat{\mu}^{(t)}=f(\hat{\mu}^{(t-1)})$ where 
\begin{equation}
	f(x)=\frac{ax}{x-b},\quad a=\frac{\hat{S}}{n},\quad b=\frac{\hat{T}}{\hat{\pi}}\left( 1-\frac{\hat{n}}{n} \right).
	\label{eq:new-notation}
\end{equation}
In this notation, the fixed point of $f$ is $x=a+b=\hat{\mu}_{\text{AN}}$.

	The proof now proceeds in two steps: first, we generalize the problem slightly and consider the discrete dynamical system $x^{(t)}$ initialized at $x^{(0)}=a$ and with iterative map $x^{(t)}=f(x^{(t-1)})$ for arbitrary fixed $a$ and $b$. For this dynamical system, we show that if $|a|>|b|$, then $x^{(t)}$ converges to $a+b$. Second, we show that, with high probability, the particular $a$ and $b$ defined in (\ref{eq:new-notation}) satisfy these conditions	.

\paragraph{Analyzing the dynamical system.} 
The function $f$ has two fixed points, at $x^*_1=0$ and at $x_2^*=a+b$. To understand the stability of these fixed points, we compute the derivative $$|f'(x)|=\left|\frac{ab}{(x-b)^2}\right|\implies |f'(x_1^*)|=\frac{|a|}{|b|},\quad |f'(x_2^*)|=\frac{|b|}{|a|}.$$ A fixed point $x^*$ of the map $f$ is stable if and only if $|f'(x^*)|<1$, so we see that if $|a|>|b|$, then $x_1^*$ is unstable and $x_2^*$ is stable, and if $|a|<|b|$, then $x_1^*$ is stable and $x_2^*$ is unstable. The case $|a|=|b|$ occurs with probability zero and so we do not consider it. 

In light of these stabilities, we should expect $x^{(t)}$ to converge to $a+b$ whenever $|a|>|b|$. The following lemma confirms this.

	\begin{lemma}
		If $|a|>|b|$, then the dynamical system with initial point $x_0=a$ and iterative map $x_{t+1}=f(x_t)$ converges to $x^*=a+b$.
	\label{lem:dynamics}
	\end{lemma}

	\begin{proof}
		Our analysis relies on the two-step map $$f(f(x))=\frac{a^2x}{x(a-b)+b^2}.$$ In particular, we will show that the subsequences $x_0, x_2,\cdots,$ and $x_1, x_3,\cdots,$ both converge to $x^*$, and this will prove the lemma. In both cases, the key observation is that
		\begin{align}
			|f(f(x))-x^*|&= \left|\frac{a^2x}{x(a-b)+b^2}-(a+b)\right| \\
			&= \left|\frac{b^2}{x(a-b)+b^2}\right|\cdot |x-(a+b)|.\label{contraction}
		\end{align}

		We first consider the subsequence $x_0, x_2,\cdots$, which for convenience we denote by $y_t=x_{2t}$. We claim that for any $t\geq 0$, we have 
		\begin{equation}
			\left|y_{t+1}-x^*\right|\leq \frac{|b|}{|a|}|y_t-x^*|,\label{induction}
		\end{equation}
The proof of the claim is by induction. 
For the base case, which is $t=0$ and $y_0=a$, we have that $$\left|\frac{b^2}{a(a-b)+b^2}\right|\leq \left|\frac{b^2}{ab}\right|=\frac{|b|}{|a|},$$ and substituting this into (\ref{contraction}) gives (\ref{induction}).

For the inductive step, suppose the result holds for $y_1,\cdots, y_{t-1}$. We will show that $y_t$ also satisfies $$\left|\frac{b^2}{y_t(a-b)+b^2}\right|\leq \frac{|b|}{|a|},$$ and this together with (\ref{contraction}) will prove the claim. 
The previous display is equivalent to 
		\begin{equation}
			|y_t(a-b)+b^2|\geq |ab|. \label{bound}
		\end{equation}
This inequality can be established by casework on the signs of $a$ and $b$. We discuss the case $a>0, b>0$ in detail; the other three cases are analogous. 

		If $a>0$ and $b>0$, then $a+b>a$, and since by the inductive hypothesis, $y_t$ is closer to $a+b$ than $a$ is, we must have $a\leq y_t\leq a+2b$. Thus $$|y_t(a-b)+b^2|\geq \min_{a\leq t\leq a+2b} |t(a-b)+b^2|.$$ Since the function we are minimizing is piecewise linear, the minimum must be attained at an endpoint ($t=a$ or $t=a+2b$) or where $t(a-b)+b^2=0$. 

		At $t=a$, the objective is $|a^2-ab+b^2|\geq |ab|$; at $t=a+2b$ the objective is $|a^2+ab-b^2|$. Since $|a|>|b|$ and $a,b$ are both positive, this is equal to $a^2+ab-b^2\geq ab$ as well. Finally, $t(a-b)+b^2=0$ is not possible because this requires $t=\frac{b^2}{b-a}$ and because $a>b$, $\frac{b^2}{b-a}<0<a$.
		Thus we conclude that (\ref{bound}) holds, and this completes the induction for the case $a>0$ and $b>0$. The other cases are analogous, and combining them establishes \eqref{induction}.

		Now, using our claim in \eqref{induction} repeatedly, we have that for any $t>0$, $$|y_t-x^*|\leq \frac{|b|}{|a|}|y_{t-1}-x^*|\leq\cdots \leq \left( \frac{|b|}{|a|} \right)^t|y_0-x^*|.$$ Since $|b|<|a|$, we thus have $|y_t-x^*|\to 0$ as $t\to \infty$, proving the convergence.

		Recalling that $y_t=x_{2t}$, we have shown that the subsequence $x_0, x_2,\cdots,$ converges to $x^*$. An analogous argument gives that $x_1, x_3,\cdots$ converges to $x^*$ as well, and these two together imply that $x_t\to x^*$.
	\end{proof}

\paragraph{High-probability guarantees}

Our next lemma carries out the proof of the second part of Theorem~\ref{thm:denom-convergence}, which is showing that $a$ and $b$ as defined in (\ref{eq:new-notation}) satisfy $|a|>|b|$ with high probability.

\begin{lemma}
	\label{lem:hoeffding}
	Suppose Assumption~\ref{bounded} holds and that $\mu\neq 0$. Then, $$\P\left( \left|\frac{\hat{S}}{n}\right|>\left|\frac{\hat{T}}{\hat{\pi}}\left( 1-\frac{\hat{n}}{n} \right)\right| \right)\geq 1-4\exp\left(\frac{-2\mu^2n}{O(M/\delta)^2}\right).$$
\end{lemma}

\begin{proof}
	For any $0\leq \epsilon\leq |\mu|$, define the events $$E_1=\left\{ \left|\frac{\hat{S}}{n}\right|> \left|\frac{\hat{T}}{\hat{\pi}}\left( 1-\frac{\hat{n}}{n} \right)\right|\right\}, E_2=\left\{ \left|\frac{\hat{S}}{n}-\mu\right| \leq \epsilon \right\}, E_3=\left\{ \left|\frac{\hat{T}}{\hat{\pi}}\left( 1-\frac{\hat{n}}{n} \right)\right|\leq |\mu|-\epsilon \right\}.$$ 

	We need to lower bound the probability of $E_1$, which we do by upper bounding the probability of $E_1^c$. By the union bound, 
	\begin{align*}
		\P(E_1^c)&\leq \P(E_2^c)+\P(E_3^c),\\
		&= \P\left( \left|\frac{1}{n}\sum_{k=1}^n \frac{Y_kI_k}{p_k}-\mu\right|>\epsilon \right)+\P\left( \left|\frac{\hat{T}}{\hat{\pi}}\left( 1-\frac{\hat{n}}{n} \right)\right|>|\mu|-\epsilon \right),\\
		&\leq\P\left( \left|\frac{1}{n}\sum_{k=1}^n \frac{Y_kI_k}{p_k}-\mu\right|>\epsilon \right)+\P\left( \left|\left( 1-\frac{\hat{n}}{n} \right)\right|> \frac{|\mu|-\epsilon}{M}\right),\\
		&\leq 2\exp\left( -\frac{2\epsilon^2n}{(M/\delta)^2} \right)+2\exp\left( -\frac{2(|\mu|-\epsilon)^2n}{(1/\delta)^2} \right).
	\end{align*}
	The second inequality follows from the bound 

	\begin{equation}
		\left|\hat{T}\right|=\left|\frac{1}{\hat{n}}\sum_{k=1}^n Y_k\frac{1-p_k}{p_k}\cdot \frac{I_k}{p_k}\right|\leq \frac{1}{\hat{n}}\sum_{k=1}^n |Y_k|\frac{1-p_k}{p_k}\cdot \frac{I_k}{p_k}\leq M|\hat{\pi}|,
		\label{eq:ratio-bound}
	\end{equation}
	which implies that $|\hat{T}/\hat{\pi}|\leq M$. The third inequality follows from applying Hoeffding's inequality to each term with the bounds $|Y_kI_k/p_k|\leq M/\delta$ and $|I_k/p_k|\leq 1/\delta$. 

	Finally, we choose $\epsilon$ to balance these two terms. The optimal choice is $\epsilon=\frac{M}{M+1}|\mu|$, and with this value of $\epsilon$ we conclude that $$\P(E_1)\geq 1 - 4\exp\left( -\frac{2\mu^2n}{(M+1)^2/\delta^2}\right),$$ finishing the proof.
\end{proof}

\paragraph{Combining the lemmas}

With these two lemmas, the proof of Theorem \ref{thm:denom-convergence} is straightforward. 

\begin{proof}[Proof of Theorem \ref{thm:denom-convergence}].
	Let $a, b$ be as defined in (\ref{eq:new-notation}).

	For (i), if $|a|>|b|$, then Lemma \ref{lem:dynamics} implies the result with $\hat{\mu}_{\text{lim}}=\hat{\mu}_{\text{fp}}$ If $|a|<|b|$, then an argument similar to the one in the proof of Lemma \ref{lem:dynamics} establishes that $x^{(t)}\to 0$ as $t\to \infty$, and so the statement holds with $\hat{\mu}_{\text{lim}}=0$.

	For (ii), we have $$\P(\hat{\mu}_{\text{lim}}\neq \hat{\mu}_{\text{fp}})\leq \P\left( \left|\frac{\hat{S}}{n}\right|>\left|\frac{\hat{T}}{\hat{\pi}}\left( 1-\frac{\hat{n}}{n} \right)\right| \right)\leq 4\exp\left( -\frac{2\mu^2n}{O(M/\delta)^2} \right),$$ where the first inequality is from the contrapositive of Lemma \ref{lem:dynamics} and the second is Lemma \ref{lem:hoeffding}. This bound goes to zero, implying (ii). 
\end{proof}

\paragraph{Intuition for variance reduction.}

In this section, we present the intuitive argument for why every two applications of $f$ reduce variance, as alluded to in Section~\ref{subsec:variance}.

For convenience, let $g(x)=f(f(x))$. If it were the case that $|g'(x)|\leq 1$ for all $x$, then applying $g$  would certainly reduce variance, and we would conclude that $\hat{\mu}_{\AN}$ has smaller variance than $\hat{\mu}_{\text{HT}}$. But this is not the case: $g(x)$ approaches $\infty$ as its denominator approaches 0, and so $|g'(x)|$ can be arbitrarily large . However, the actual sequence of iterates $\hat{\mu}^{(2t)}$ lies, with high probability, in an interval where $|g'(x)|$ is bounded by 1.

In what follows, we assume again that $\mu\neq 0$, and that $|a|>2|b|$, which holds with high probability by the same arguments as those in the proof of Lemma \ref{lem:hoeffding}. We also assume, without the loss of generality, that $a$ and $b$ are both positive, but this is only for clarity. Now, in the proof of Lemma \ref{lem:dynamics}, we showed that under these conditions, the entire sequence of iterates $\hat{\mu}^{(0)},\hat{\mu}^{(2)},\cdots$ lies in the interval $[a, a+2b]$. By standard concentration arguments, $a$ is concentrated around $\mu$ and $b$ is concentrated around $0$, so the interval between $[a, a+2b]$ lies with high probability in a small interval $\mathcal{I}$ centered at $\mu$. On the other hand, we can check by direct computation that $|g'(x)|\leq 1$ for $x$ outside the interval $(-3b, b)$. Again by concentration arguments, for large $n$, this interval will be concentrated around zero, and thus be disjoint from the interval $\mathcal{I}$ centered at $\mu$. 

To summarize, there exists an interval $\mathcal{I}$ centered around $\mu$ such that, with high probability, $x^{(2t)}\in \mathcal{I}$ for all $t>0$ and $|g'(x)|\leq 1$ for all $x\in\mathcal{I}$. If the function $g$ were fixed, this would be enough to conclude that each application of $g$ reduces variance, and thus that $\hat{\mu}_{\AN}$ has smaller variance than $\hat{\mu}_{\text{HT}}$. Unfortunately, because the function $g$ is random, this is not a rigorous argument, only an intuitive one. 

We note that an argument similar to ours appears in Section 5 of \cite{gmm}. In that context, when studying a generalized method of moments (GMM) estimator, the authors iterate a random data-dependent estimate of an underlying true function. The underlying true function is a contraction mapping, and thus reduces variance when applied to a random variable, so they argue heuristically that the iterations of the data-dependent estimated function should reduce variance as well.

\subsection{Proof of Theorem \ref{thm:aipw-equivalence}}
\label{subsec:equivalence}

\begin{proof}
	Let $\mathcal{T}_n$ be an auxiliary data set of size $n$ on which $\hat{\mu}(\cdot)$ and $\hat{p}(\cdot)$ are trained. (This will be the case if, for example, we do cross-fitting as in \cite{dml}; we choose not to write out the explicit cross-fitting set-up to simplify the notation.) Then, it is sufficient to show that the correction term we have introduced is $o_P(n^{-1/2})$. That error term is 
	\begin{align}
		&=\frac{1}{\hat{\pi}}\left( \frac{1}{n}\sum_{k=1}^n(Y_k-\hat{\mu}(X_k))\frac{1-\hat{p}(X_k)}{\hat{p}(X_k)}\frac{I_k}{\hat{p}(X_k)} \right)\left( 1-\frac{1}{n}\sum_{k=1}^n \frac{I_k}{\hat{p}(X_k)} \right),\\
		&\lesssim \left( \frac{1}{n}\sum_{k=1}^n (Y_k-\hat{\mu}(X_k))I_k \right)\left( 1-\frac{1}{n}\sum_{k=1}^n \frac{I_k}{\hat{p}(X_k)} \right),\\
		&=\left( \frac{1}{n}\sum_{k=1}^n (Y_k-\mu(X_k))I_k+\frac{1}{n}\sum_{k=1}^n (\mu(X_k)-\hat{\mu}(X_k))I_k \right)\left( 1-\frac{1}{n}\sum_{k=1}^n \frac{I_k}{p(X_k)}+\frac{1}{n}\sum_{k=1}^n \frac{I_k}{p(X_k)}-\frac{I_k}{\hat{p}(X_k)} \right),\\
		&\leq \left( O_P(n^{-1/2})+\frac{1}{n}\sum_{k=1}^n |(\mu(X_k)-\hat{\mu}(X_k))I_k| \right)\left( O_P(n^{-1/2})+\frac{1}{n}\sum_{k=1}^n \left|\frac{I_k}{p(X_k)}-\frac{I_k}{\hat{p}(X_k)}\right| \right),\\
		&\leq 
		\begin{multlined}[t]
			O_P(n^{-1})+\frac{O_P(n^{-1/2})}{n}\sum_{k=1}^n |\mu(X_k)-\hat{\mu}(X_k)|+\frac{O_P(n^{-1/2})}{n}\sum_{k=1}^n \left|\frac{1}{p(X_k)}-\frac{1}{\hat{p}(X_k)}\right|+\\ \left( \frac{1}{n}\sum_{k=1}^n |\mu(X_k)-\hat{\mu}(X_k|\right)\left( \frac{1}{n}\sum_{k=1}^n \left|\frac{1}{p(X_k)}-\frac{1}{\hat{p}(X_k)}\right|\right),
			\end{multlined}\\
			&\leq 
			\begin{multlined}[t]
				O_P(n^{-1})+O_P(n^{-1/2})\left( \sup_{x\in\mathcal{X}} |\mu(x)-\hat{\mu}(x)|+\sup_{x\in\mathcal{X}}\left|\frac{1}{p(x)}-\frac{1}{\hat{p}(x)}\right| \right)+\\ \left( \frac{1}{n}\sum_{k=1}^n |\mu(X_k)-\hat{\mu}(X_k|\right)\left( \frac{1}{n}\sum_{k=1}^n \left|\frac{1}{p(X_k)}-\frac{1}{\hat{p}(X_k)}\right|\right)\\
			\end{multlined}\\
			&= O_P(n^{-1})+O_P(n^{-1/2})o_P(1)+\left( \frac{1}{n}\sum_{k=1}^n |\mu(X_k)-\hat{\mu}(X_k|\right)\left( \frac{1}{n}\sum_{k=1}^n \left|\frac{1}{p(X_k)}-\frac{1}{\hat{p}(X_k)}\right|\right).\label{eq:decay-rate}
	\end{align}
where the second inequality follows from the fact that, by the consistency of $\hat{p}(\cdot)$, we have $\delta/2\leq \hat{p}(x)\leq 1-\delta/2$ for all $x\in\mathcal{X}$ for sufficiently large $n$, the fourth inequality applies the triangle inequality and the CLT for i.i.d. sums, and the final equality uses Assumption~\ref{consistency}. (Note that since $p(\cdot)$ is bounded, the consistency of $\hat{p}$ implies the consistency of $1/\hat{p}$ as well.)

	Examining \eqref{eq:decay-rate}, the first two terms are $o_P(n^{-1/2})$ as needed, so it remains to analyze the final term. We thus compute \[\E\left[ \left( \frac{1}{n}\sum_{k=1}^n |\mu(X_k)-\hat{\mu}(X_k|\right)^2\left( \frac{1}{n}\sum_{k=1}^n \left|\frac{1}{p(X_k)}-\frac{1}{\hat{p}(X_k)}\right|\right)^2 \right]\] as 

	\begin{align}
		&\leq \E\left[ \left( \frac{1}{n}\sum_{k=1}^n (\mu(X_k)-\hat{\mu}(X_k))^2 \right)\left( \frac{1}{n}\sum_{k=1}^n \left( \frac{1}{p(X_k)}-\frac{1}{\hat{p}(X_k)} \right)^2 \right) \right],\\
		&= \E\left[ \frac{1}{n^2}\sum_{k\neq \ell}( \mu(X_k)-\hat{\mu}(X_k))^2\left( \frac{1}{p(X_{\ell})}-\frac{1}{\hat{p}(X_{\ell})} \right)^2\right]+o(n^{-1}),\\
		&= \E\left[ \frac{1}{n^2}\sum_{k\neq \ell} \E\left[ (\mu(X_k)-\hat{\mu}(X_k))^2\left( \frac{1}{p(X_{\ell})}-\frac{1}{\hat{p}(X_{\ell})} \right)^2\mid \mathcal{T}_n \right]\right]+o(n^{-1}),\\
		&= \E\left[ \frac{n(n-1)}{n^2} \E\left[ (\mu(X_k)-\hat{\mu}(X_k))^2\left( \frac{1}{p(X_{\ell})}-\frac{1}{\hat{p}(X_{\ell})} \right)^2\mid \mathcal{T}_n \right]\right]+o(n^{-1}),\\
		&\leq \E\left[\E\left[ (\mu(X_k)-\hat{\mu}(X_k))^2\mid \mathcal{T}_n \right]\E\left[ \left( \frac{1}{p(X_{\ell})}-\frac{1}{\hat{p}(X_{\ell})} \right)^2\mid \mathcal{T}_n \right]\right]+o(n^{-1}),\label{eq:product}
	\end{align}
where the first inequality is Cauchy-Schwarz, the second equality expands the product and drops the $k=\ell$ terms, the fourth equality uses the fact that the terms of the sum are equal after conditioning on $\mathcal{T}_n$, and the fifth equality uses the fact that the two errors are independent after conditioning on $\mathcal{T}_n$.

	By Assumption~\ref{risk}, the product of conditional expectations in \eqref{eq:product} is $o_P(n^{-1})$. Since $\mu(X_k)$ and $p(X_{\ell})$ are bounded by Assumption~\ref{bounded}, and $\hat{\mu}(\cdot)$ and $\hat{p}(\cdot)$ are sup-norm consistent by Assumption~\ref{consistency}, that product of conditional expectations is eventually dominated by a constant, and so the expectation in \eqref{eq:product} is $o(n^{-1})$ as well, completing the proof.
\end{proof}

\subsection{Proof of Theorem \ref{thm:regret-bound}}

Our proof closely follows ideas and tools developed in \cite{kitagawa}.

\begin{proof}
	Our goal is to control $V(\pi^*)-V(\hat{\pi}_{\text{AN}})$. We do this by first writing $V(\pi^*)-V(\hat{\pi}_{\text{AN}})$ as

	\begin{align}
		&= V(\pi^*)-\hat{V}_{\text{AN}}(\hat{\pi}_{\text{ AN}})+\hat{V}_{\text{AN}}(\hat{\pi}_{\text{AN}})-V(\hat{\pi}_{\text{AN}}),\\
	&\leq V(\pi^*)-\hat{V}_{\text{AN}}(\pi^*)+\sup_{\pi \in \Pi}|\hat{V}_{\text{AN}}(\pi)-V(\pi)|,\\
		&\leq 2\sup_{\pi \in \Pi} |\hat{V}_{\text{AN}}(\pi)-V(\pi)|,\\
		&\leq 2\sup_{\pi \in \Pi} |\hat{V}_{\text{AN}}(\pi)-\hat{V}_{\text{IPW}}(\pi)|+2\sup_{\pi\in \Pi}|\hat{V}_{\text{IPW}}(\pi)-V(\pi)|.\label{eq:decompose}
	\end{align}

	The second term of \eqref{eq:decompose} is bounded in Theorem 2.1 of \cite{kitagawa}, and so we are interested in bounding the first term. By the definitions of $\hat{V}_{\text{IPW}}$ and $\hat{V}_{\text{AN}}$, that first term of \eqref{eq:decompose} is 
	\begin{align}
		&=\sup_{\pi\in\Pi}\left|\frac{\sum_{k=1}^n Y_k\frac{1-\P(I_k=\pi(X_k)\mid X_k)}{\P(I_k=\pi(X_k)\mid X_k)}\frac{\mathbf{1}\left\{ I_k=\pi(X_k) \right\}}{\P(I_k=\pi(X_k)\mid X_k)}}{\sum_{k=1}^n\frac{1-\P(I_k=\pi(X_k)\mid X_k)}{\P(I_k=\pi(X_k)\mid X_k)}\frac{\mathbf{1}\left\{ I_k=\pi(X_k) \right\}}{\P(I_k=\pi(X_k)\mid X_k)}}\left( 1-\frac{1}{n}\sum_{k=1}^n \frac{\mathbf{1}\left\{ I_k=\pi(X_k) \right\}}{\P(I_k=\pi(X_k)\mid X_k)}\right)\right|,\\
		&\leq M\sup_{\pi \in \Pi}\left|1-\frac{1}{n}\sum_{k=1}^n \frac{\mathbf{1}\left\{ I_k=\pi(X_k) \right\}}{\P(I_k=\pi(X_k)\mid X_k)}\right|, \label{eq:diff-bound}
	\end{align}
by a calculation analogous to the one in \eqref{eq:ratio-bound}. Combining \eqref{eq:decompose} and \eqref{eq:diff-bound} gives

	\begin{equation}
		\E[V(\pi^*)-V(\hat{\pi}_{\text{AN}})]\leq M\E\left[ \sup_{\pi \in \Pi} \left|1-\frac{1}{n}\sum_{k=1}^n \frac{\mathbf{1}\left\{ I_k=\pi(X_k) \right\}}{\P(I_k=\pi(X_k)\mid X_k)}\right|\right]+O\left( \frac{M}{\delta}\sqrt{\frac{\text{VC}(\Pi)}{n}} \right).
		\label{eq:process}
	\end{equation}

	Finally, the expectation in the first term of \eqref{eq:process} can be bounded using standard empirical process tools; in particular, Lemmas A.1 and A.4 of \cite{kitagawa} imply that it is $O\left( \frac{1}{\delta}\sqrt{\frac{\text{VC}(\Pi)}{n}} \right)$, finishing the proof.
\end{proof}

\section{Further experiments}
\label{sec:sims_app}
\subsection{Semi-synthetic experiments}

We now move beyond pure simulations and consider a real data set. Specifically we work with a dataset of Swiss municipalities, provided by the R package \texttt{sampling} \cite{sampling} under the GPL-2 license. This data set contains demographic and financial data for 2896 different municipalities in Switzerland, all aggregated at the municipality level, and does not contain data about any individual. We consider two responses: $Y_1$, the wooded area of a municipality and $Y_2$, the industrial area of a municipality, and we assume a probability-proportional-to-size sampling scheme in which the probability of observing the response is proportional to the total area of the municipality. We take the sum of the probabilities to be either 50 or 250, meaning that the average number of observations is either 50 or 250. We perform simulations by resampling the set of municipalities that are observed according to the given probabilities and comparing the estimators to the mean value across all municipalities, making this a fixed-population setting. The RMSE of the Horvitz--Thompson estimator, \Hajek{} estimator, and adaptively normalized estimator for the 4 specifications are shown in Table \ref{tab:swiss}.

\begin{table}[t]
	\centering
	\begin{tabular}{lllll}
		\toprule
		&\multicolumn{4}{c}{\emph{Problem specification}}\\
		\cmidrule{2-5}
		\emph{Estimator}&$\sum p_k=50, Y_1$ & $\sum p_k=250, Y_1$& $\sum p_k=50, Y_2$& $\sum p_k= 250, Y_2$\\
		\midrule
		$\hat{\mu}_{\text{HT}}$&$68.4\pm 0.1030$&$27.8\pm 0.0710$&$2.51\pm 0.0051$&$1.07\pm 0.0026$\\
		$\hat{\mu}_{\Hajek{}}$&$95.3\pm 0.3587$&$39.3\pm 0.1510$&$2.52\pm 0.0076$&$1.06\pm 0.00244$\\
		$\hat{\mu}_{\text{AN}}$&$61.5\pm 0.1035$&$23.1\pm 0.0538$&$2.45\pm 0.0086$&$1.01\pm0.0028$\\
		\bottomrule
	\end{tabular}
	\caption{RMSE of estimators on Swiss municipality data; $Y_1$ is wood area and $Y_2$ is industrial area; probabilities are chosen proportional to total municipality area, which is strongly positively correlated with $Y_1$ and weakly positively correlated with $Y_2$. RMSEs are averaged over 100,000 trials and standard errors are over 10 replications.}
	\label{tab:swiss}
\end{table}

To contextualize the results, we note that $Y_1$ is strongly positively correlated with the probabilities and $Y_2$ is weakly positively correlated with them. Thus in the first two columns of Table \ref{tab:swiss}, the Horvitz--Thompson estimator is much better than the \Hajek{} estimator, but not as good as the adaptively normalized estimator. In the latter two columns, because the correlation is weaker, the Horvitz--Thompson and \Hajek{} estimators have nearly the same RMSE, but again the adaptively normalized estimator improves on both. For reference, plots of the MSE in each problem as a function of different values of $\lambda$ are shown in Figure~\ref{fig:swiss-opt}.

\begin{figure}[t]
	\centering
	\begin{tabular}{cc}
		\includegraphics[width=0.4\linewidth]{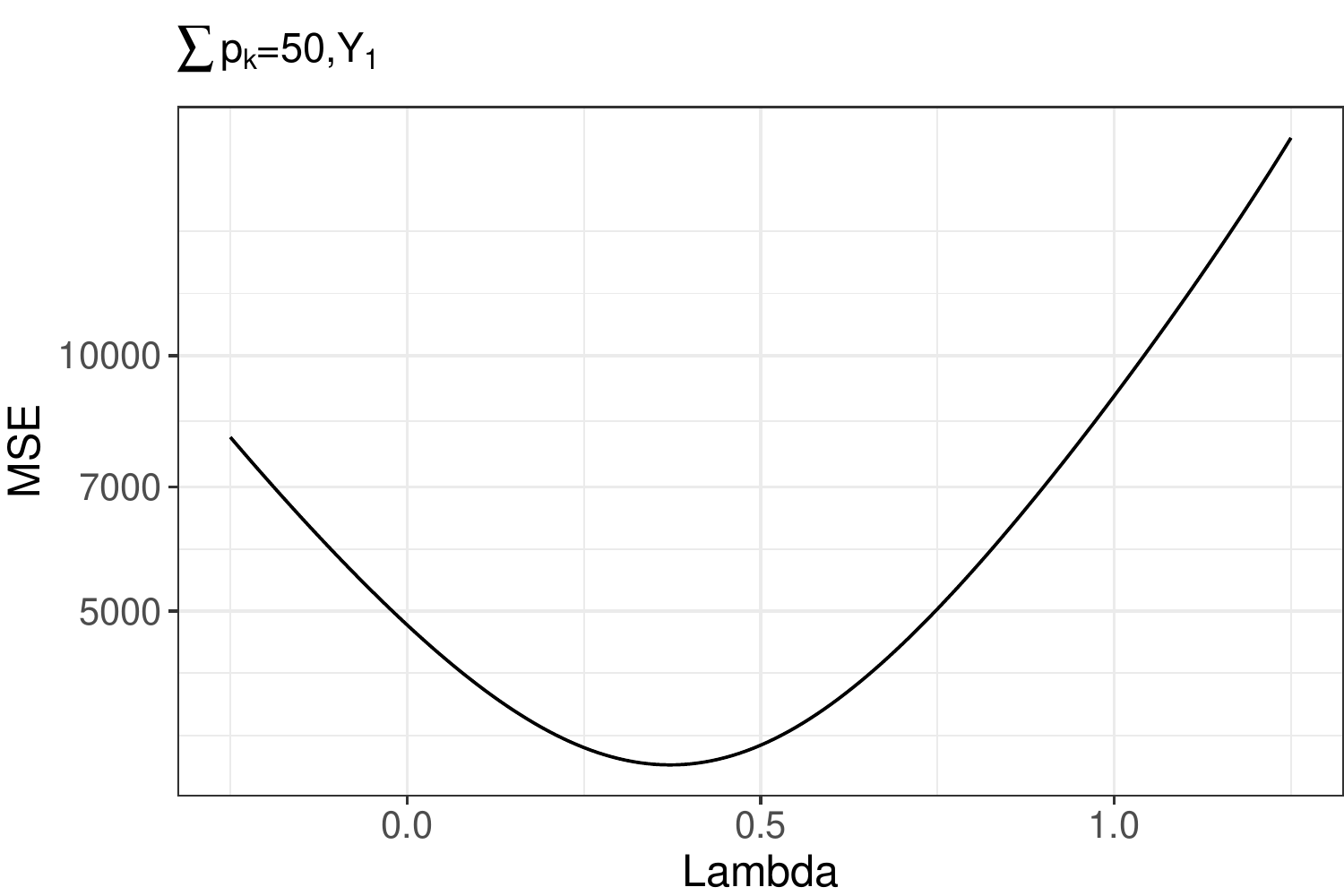}&\includegraphics[width=0.4\linewidth]{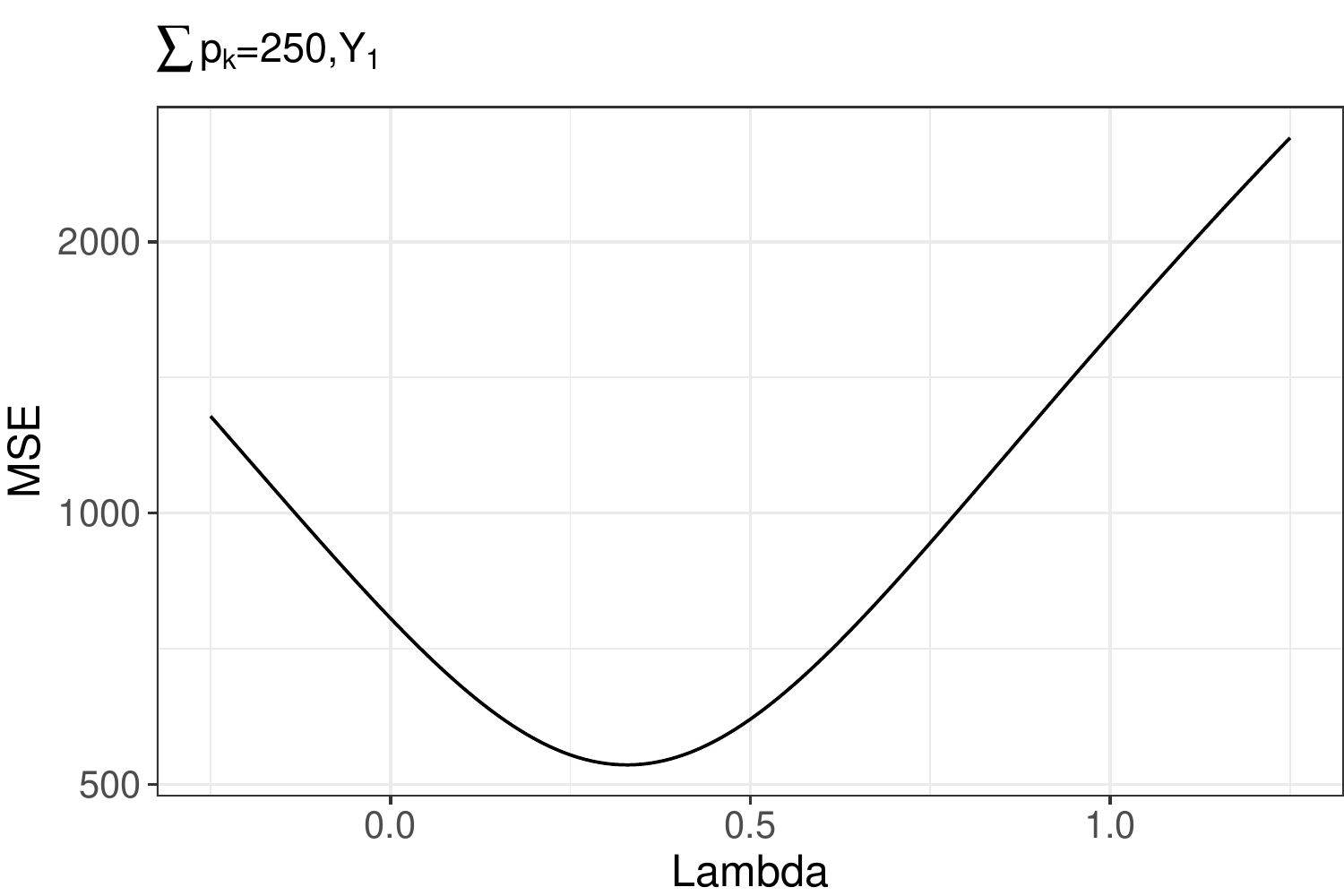}\\
	\vspace{-2mm}
		\includegraphics[width=0.4\linewidth]{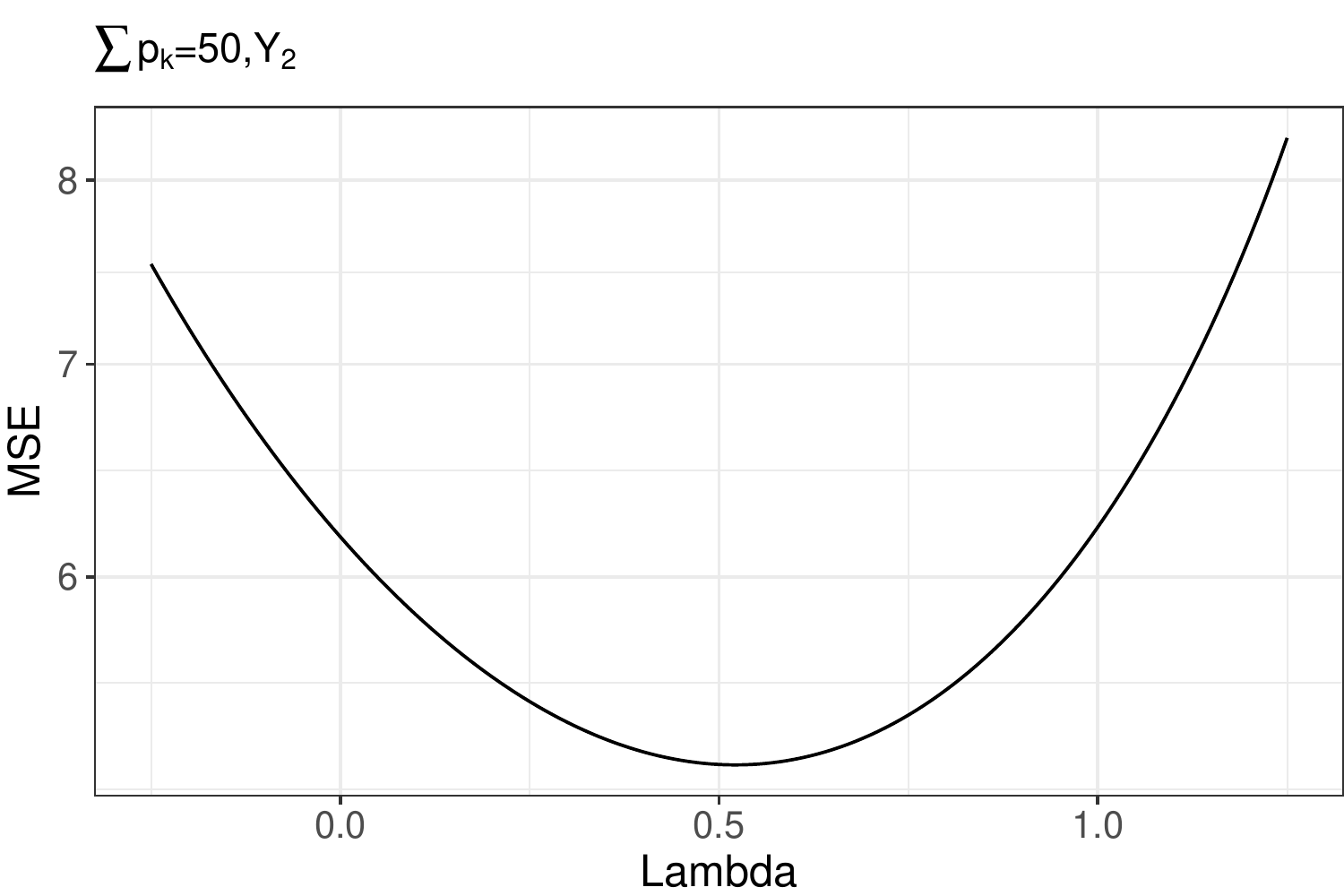}&\includegraphics[width=0.4\linewidth]{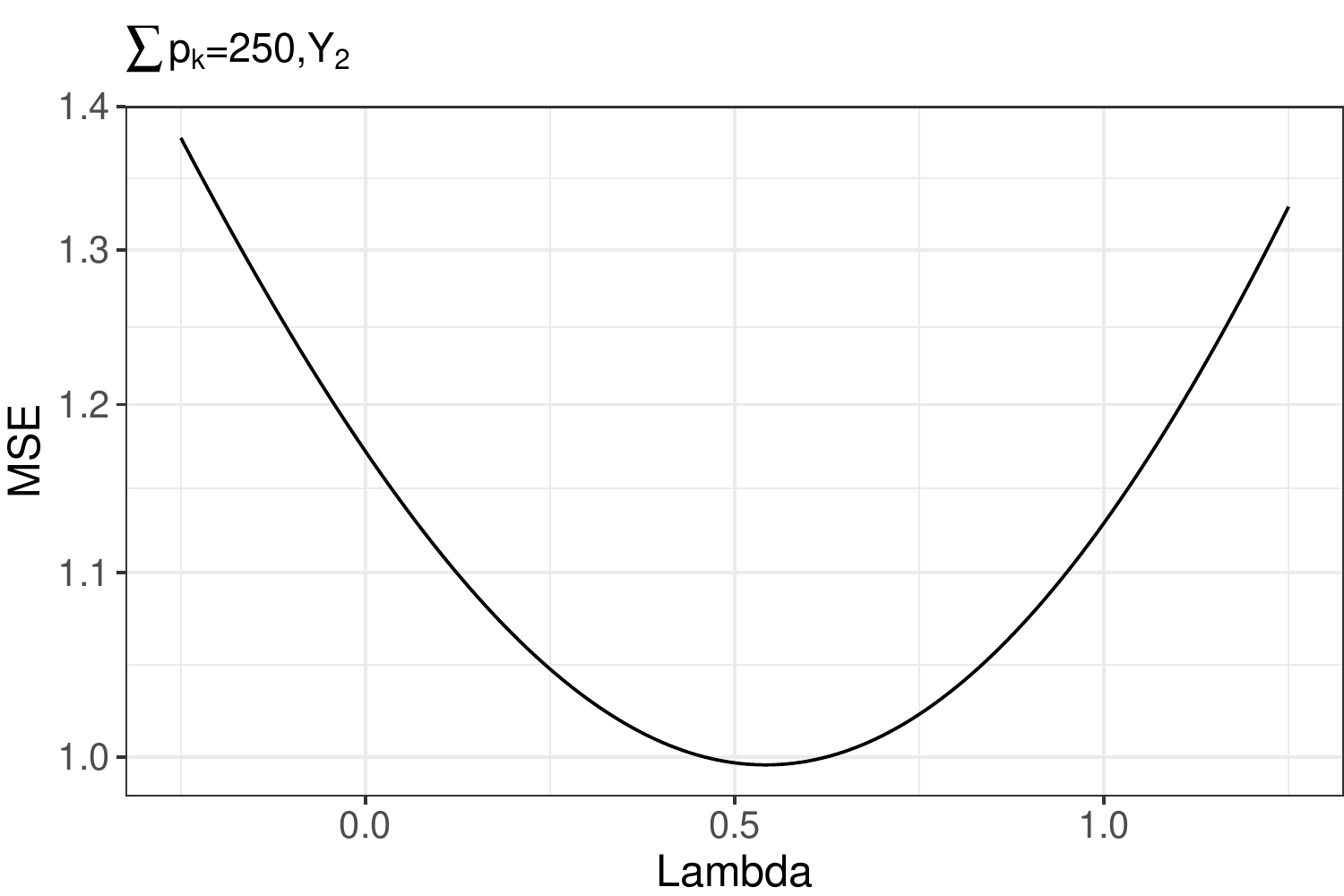}\\
	\end{tabular}
	\caption{MSE of $\hat{\mu}_{\lambda}$ as a function of $\lambda$ for the four problem specifications of the Swiss municipality data. Interestingly, the optimal choices of $\lambda$ are all between 0 and 1, suggesting that the \Hajek{} estimator is actually ``over-normalizing'' in this case.}
	\label{fig:swiss-opt}
\end{figure}

\subsection{Confidence intervals and coverage}

In this section we discuss the problem of constructing confidence intervals for $\hat{\mu}_{\AN}$. Since we have computed the asymptotic variance of $\hat{\mu}_{\AN}$ in Theorem~\ref{thm:mu-hat-fp-clt}, we can use a normal approximation to construct confidence intervals that are asymptotically valid. The coverage of these intervals in a simulation experiment is shown in Figure~\ref{fig:coverage}.

\begin{figure}[t]
	\centering
	\begin{tabular}{cc}
		\includegraphics[width=0.45\linewidth]{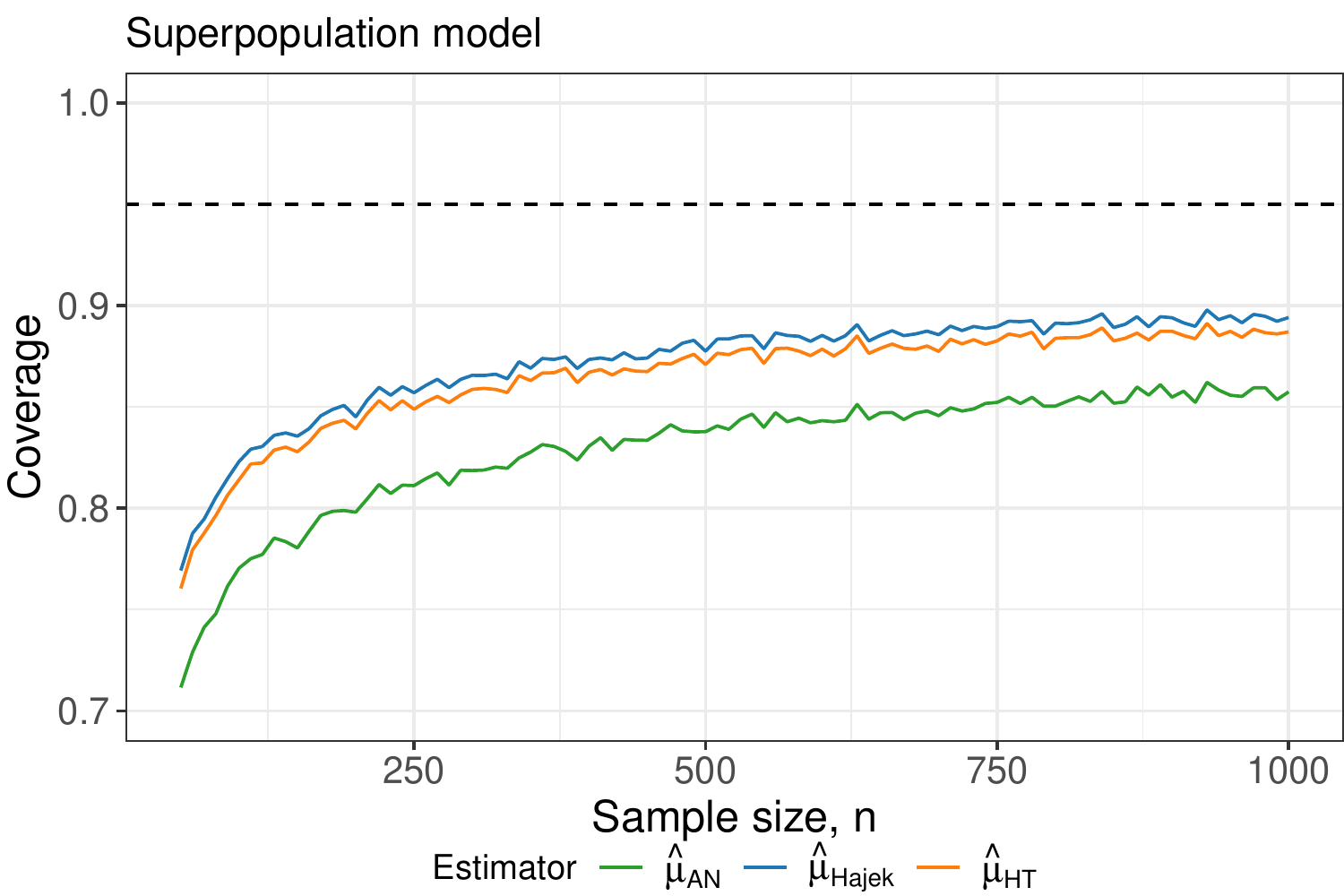}&\includegraphics[width=0.45\linewidth]{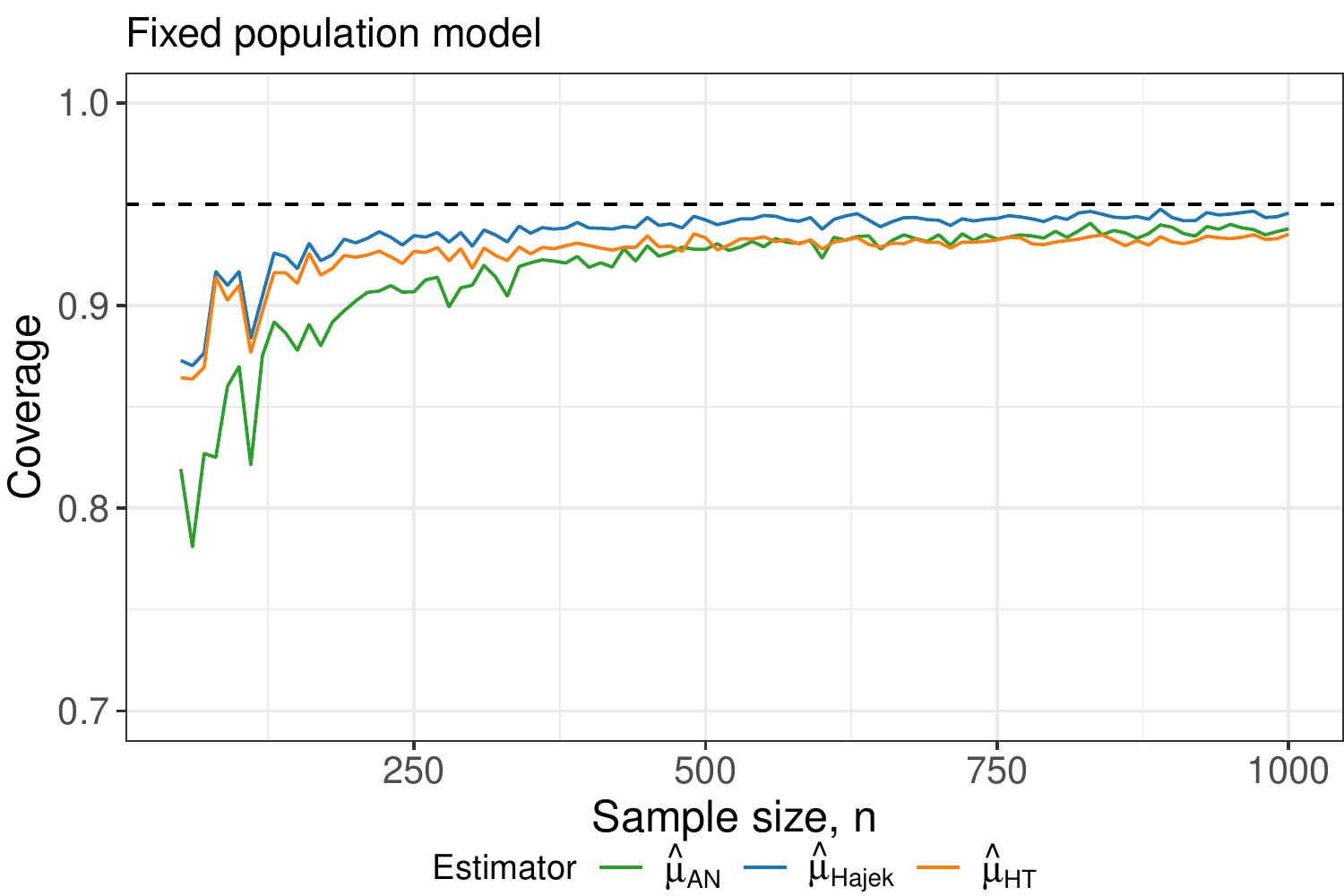}
\end{tabular}
\caption{The coverage of a 95\% confidence interval for $\hat{\mu}_{\text{HT}}$, $\hat{\mu}_{\Hajek{}}$, and $\hat{\mu}_{\AN}$ based on the asymptotic variances in Theorems~\ref{thm:mu-hat-lambda-clt} ($\lambda=0,1$) and \ref{thm:mu-hat-fp-clt} in a superpopulation model (left) and a fixed population model (right). In the superpopulation model, at each trial, the $Y_k, p_k, I_k$ are all drawn anew, while in the fixed population model, the $Y_k, p_k$ are fixed and the $I_k$ are drawn anew at each trial. The target coverage is indicated with a dashed line. The coverage in the fixed population model is significantly better than the coverage in the superpopulation model, for reasons discussed in the main text.}
	\label{fig:coverage}
\end{figure}

Unfortunately, we see that in the superpopulation model, which is the model we have considered throughout, the coverage is quite poor. To better understand the reason for this behavior, consider the asymptotic variance of the \Hajek{} estimator, $\E\left[ \frac{1-p_k}{p_k}(Y_k-\mu)^2 \right]$. Crucially, $p\mapsto \frac{1-p}{p}$ is a decreasing function of $p\in[0,1]$, and so this expectation up-weights small values of $p_k$. But these are exactly the values we are least likely to observe, hence the difficulty in obtaining an accurate approximation of the asymptotic variance, and thus in obtaining confidence intervals with good coverage. In particular, we do not believe that these results are due to a failure of the normal approximation: the normal approximation should be very nearly exactly correct for values as large as $n=1000$, but estimating the parameters of the target normal approximation remains challenging. This is further evidenced by the fact that $\hat{\mu}_{\text{HT}}$, for which there is no approximation and the asymptotic variance is exactly correct, actually has slightly lower coverage than $\hat{\mu}_{\text{\Hajek}}$. Rather, the challenge is in estimating $\E\left[ \frac{1-p_k}{p_k}Y_k^2 \right]$, which is slightly harder than estimating the asymptotic variance of $\hat{\mu}_{\Hajek}$ because $Y_k^2$ is typically larger than $(Y_k-\mu)^2$.

This point is further demonstrated by the significantly improved coverage of the same intervals in the fixed population simulation in the right panel of Figure~\ref{fig:coverage}. In this setting, the asymptotic variance of, for example, the \Hajek{} estimator is no longer $\E\left[ \frac{1-p_k}{p_k}(Y_k-\mu)^2 \right]$ but is rather $\lim \frac{1}{n}\sum_{k=1}^n \frac{1-p_k}{p_k}(Y_k-\overline{Y})^2$ where the limit is over a sequence of finite populations of increasing size and $\overline{Y}=\frac{1}{n}\sum_{k=1}^n Y_k$. The improved coverage is driven by the fact that now, only the small values of $p_k$ that are actually in the finite population, rather than all possible values of $p_k$, contribute to the variance, and so there is less unobserved variance negatively impacting the coverage.

To summarize, intervals constructed using our theory show sub-par coverage in simulation, but this is largely due to unobserved variance in the problem and the difficulty of estimating such variance, as evidenced by the comparison with the finite population setting, rather than a failure of our theory or of the normal approximation.


\section{Joint adaptive normalization in ATE estimation}
\label{sec:joint}

This appendix explores some subtleties of how the adaptive normalizations should be chosen for ATE estimation. In particular, the estimator in~\eqref{eq:ate-an} is equivalent to selecting $\lambda$ in \eqref{eqn:tt} to separately minimize the asymptotic variance of the two mean estimates. However, this ``plug-in'' use of adaptive normalization ignores the fact that the asymptotic variance of an adaptively normalized ATE estimator depends not just on the variances of the two mean estimators but also their covariance.
 
In what follows, we jointly choose the normalizations of each mean estimator to minimize the asymptotic variance of estimating the combined estimand $\tau$. Specifically, we define $$\hat{\mu}_{1, \lambda_1}=\frac{\hat{S}_1}{(1-\lambda_1)n+\lambda_1 \hat{n}_1},\quad \hat{\mu}_{0, \lambda_0}=\frac{\hat{S}_0}{(1-\lambda_0)n+\lambda_0\hat{n}_0},$$ where $$\hat{S}_1=\sum_{k=1}^n \frac{Y_k(1)I_k}{p_k},\quad \hat{S}_0=\sum_{k=1}^n \frac{Y_k(0)(1-I_k)}{1-p_k},\quad \hat{n}_1=\sum_{k=1}^n \frac{I_k}{p_k},\quad \hat{n}_0=\sum_{k=1}^n \frac{1-I_k}{1-p_k}.$$

Combining these estimators leads to the estimator $\hat{\tau}_{\lambda_1, \lambda_0}=\hat{\mu}_{1,\lambda_1}-\hat{\mu}_{0, \lambda_0}$ of $\tau$. To follow the program of Section \ref{sec:adapt}, we seek to choose $\lambda_1$ and $\lambda_0$ to minimize the following asymptotic variance.

\begin{theorem}
	\label{thm:ate-clt}
	Assuming that $Y_k(1)$ and $Y_k(0)$ both satisfy Assumption \ref{bounded}, we have $$\sqrt{n}(\hat{\tau}_{\lambda_1, \lambda_0}-\tau)\xrightarrow{d}N(0, \sigma^2),$$ where 
$$
\sigma^2=\lambda_1^2\mu_1^2\pi_1+\lambda_0^2\mu_0^2\pi_0-2\lambda_1\mu_1(T_1+\mu_0)+2\lambda_0\mu_0(-\mu_1-T_0)+2\lambda_1\lambda_0\mu_1\mu_0+C,
$$ 
for 
$$
\pi_1=\E\left[ \frac{1-p_k}{p_k} \right], \pi_0=\E\left[ \frac{p_k}{1-p_k} \right], T_1=\E\left[ Y_k(1)\frac{1-p_k}{p_k} \right], T_0=\E\left[ Y_k(0)\frac{p_k}{1-p_k} \right],
$$ 
and $C$ denotes terms that do not depend on either $\lambda_1$ or $\lambda_0$.
\end{theorem}

As with our other CLTs, this is a routine delta method calculation.

\begin{proof}
	Define the vector $$\hat{\beta}=\left( \frac{1}{n}\sum_{k=1}^n \frac{Y_k(1)I_k}{p_k}, \frac{1}{n}\sum_{k=1}^n \frac{Y_k(0)(1-I_k)}{1-p_k}, \frac{1}{n}\sum_{k=1}^n \frac{I_k}{p_k}, \frac{1}{n}\sum_{k=1}^n \frac{1-I_k}{1-p_k} \right)$$ with mean $\beta=(\mu_1, \mu_0, 1, 1)$. By the same arguments as in \ref{lem:clt}, $\hat{\beta}$ satisfies the usual CLT for the mean of i.i.d. random variables, and so we can apply the delta method with the function $f(x,y,z,w)=\frac{x}{1-\lambda_1+\lambda_1z}-\frac{y}{1-\lambda_0+\lambda_0w}.$ The relevant gradient is $(1, -1, -\lambda_1\mu_1, \lambda_0\mu_0)$ and so the asymptotic variance of $\hat{\tau}_{\lambda_1, \lambda_0}$ is 
	\begin{align*}
		&\lambda_1^2\mu_1^2\var\left( \frac{I_k}{p_k} \right)+\lambda_0^2\mu_0^2\var\left( \frac{1-I_k}{1-p_k} \right)\\
		&-2\lambda_1\mu_1\left( \cov\left( \frac{I_k}{p_k}, \frac{Y_k(1)I_k}{p_k}\right)-\cov\left( \frac{I_k}{p_k}, \frac{Y_k(0)(1-I_k)}{1-p_k} \right) \right),\\
		&+2\lambda_0\mu_0\left( \cov\left( \frac{1-I_k}{1-p_k}, \frac{Y_k(1)I_k}{p_k} \right)-\cov\left( \frac{1-I_k}{1-p_k}, \frac{Y_k(0)(1-I_k)}{1-p_k} \right) \right)\\
		&-2\lambda_1\lambda_0\mu_1\mu_0\cov\left( \frac{I_k}{p_k}, \frac{1-I_k}{1-p_k} \right) +C,
	\end{align*}

	where $C$ denotes terms that do not depend on either $\lambda_1$ or $\lambda_2$.

	Finally, we can compute $$\var\left( \frac{I_k}{p_k} \right)= \pi_1,\quad \var\left( \frac{1-I_k}{1-p_k} \right)= \pi_0$$ and $$\cov\left( \frac{I_k}{p_k}, \frac{Y_k(1)I_k}{p_k} \right)= T_1,\quad \cov\left( \frac{I_k}{p_k}, \frac{Y_k(0)(1-I_k)}{1-p_k} \right)=-\mu_0,$$ and $$\cov\left( \frac{1-I_k}{1-p_k}, \frac{Y_k(1)I_k}{p_k}\right)=-\mu_1,\quad \cov\left( \frac{1-I_k}{1-p_k}, \frac{Y_k(0)(1-I_k)}{1-p_k} \right) T_0,$$ and finally $\cov\left( \frac{I_k}{p_k}, \frac{1-I_k}{1-p_k} \right)=-1$. Substituting these in gives the result. 
\end{proof}

We now minimize the asymptotic variance of Theorem \ref{thm:ate-clt} in $\lambda_1$ and $\lambda_0$. The first-order stationary conditions tell us that the optimal pair $(\lambda_1^*, \lambda_0^*)$ will satisfy $$2\lambda_1^*\mu_1^2\pi_1-2\mu_1(T_1+\mu_0)+2\lambda_0^*\mu_1\mu_0=0\implies \lambda_1^*=\frac{T_1+\mu_0-\lambda_0^*\mu_0}{\mu_1\pi_1}$$ and $$2\lambda_0^*\mu_0^2\pi_0-2\mu_0(\mu_1+T_0)+2\lambda_1^*\mu_1\mu_0=0\implies\lambda_0^*=\frac{T_0+\mu_1-\lambda_1^*\mu_1}{\mu_0\pi_0}.$$

As before, we replace $T_0, T_1, \pi_0, \pi_1$ with IPW estimates $\hat{T}_0, \hat{T}_1, \hat{\pi}_0, \hat{\pi}_1$, and then jointly estimate $\lambda_0, \lambda_1, \mu_0, \mu_1$ by solving the system of fixed point equations 

\begin{equation}
	\hat{\lambda}_{1,\AN}=\frac{\hat{T}_1+\hat{\mu}_{0,\AN}-\hat{\lambda}_{0, \AN}\hat{\mu}_{0,\AN}}{\hat{\mu}_{1,\AN}\hat{\pi}_1},\quad \hat{\lambda}_{0,\AN}=\frac{\hat{T}_0+\hat{\mu}_{1,\AN}-\hat{\lambda}_{1,\AN}\hat{\mu}_{1,\AN}}{\hat{\mu}_{0,\AN}\hat{\pi}_0},
	\label{eq:lambda-fp}
\end{equation}
and 
\begin{equation}
	\hat{\mu}_{1,\AN}=\frac{\hat{S}_1}{(1-\hat{\lambda}_{1,\AN})n+\hat{\lambda}_{1,\AN}\hat{n}_1},\quad \hat{\mu}_{0,\AN}=\frac{\hat{S}_0}{(1-\hat{\lambda}_{0,\AN})n+\hat{\lambda}_{0,\AN}\hat{n}_0}.
	\label{eq:mu-fp}
\end{equation}

(Note that we have slightly overloaded the notation $\hat{\mu}_{1,\AN}$, which is used differently in Section~\ref{subsec:ate}. In the entirety of this appendix, the definition above is the one used.)

To actually solve this system, we focus first on \eqref{eq:lambda-fp}, noting that it is linear in $\hat{\lambda}_{1,\AN}\hat{\mu}_{1,\AN}$ and $\hat{\lambda}_{0,\AN}\hat{\mu}_{0,\AN}$, and has solution
\begin{equation}
	\hat{\lambda}_{1,\AN}\hat{\mu}_{1,\AN}=\frac{\hat{T}_1\hat{\pi}_0+\hat{\mu}_{0,\AN}\hat{\pi}_0-\hat{T}_0-\hat{\mu}_{1,\AN}}{\hat{\pi}_0\hat{\pi}_1-1},\quad \hat{\lambda}_{0,\AN}\hat{\mu}_{0,\AN}=\frac{\hat{T}_0\hat{\pi}_1+\hat{\mu}_{1,\AN}\hat{\pi}_1-\hat{T}_1-\hat{\mu}_{0,\AN}}{\hat{\pi}_0\hat{\pi}_1-1}.
	\label{eq:linear-sol}
\end{equation}
Then, we rewrite \eqref{eq:mu-fp} as 
\begin{equation}
	\hat{\mu}_{1,\AN}=\hat{\mu}_{1,\AN}\hat{\lambda}_{1,\AN}\left(1-\frac{\hat{n}_1}{n}\right)+\frac{\hat{S}_1}{n},\quad \hat{\mu}_{0,\AN}=\hat{\mu}_{0,\AN}\hat{\lambda}_{0,\AN}\left( 1-\frac{\hat{n}_0}{n} \right)+\frac{\hat{S}_0}{n}.
\end{equation}
Using \eqref{eq:linear-sol}, we can conclude that $\hat{\mu}_{1,\AN}$ and $\hat{\mu}_{0,\AN}$ must satisfy the system of linear equations $$\left[ \begin{array}{cc}1+\frac{1}{\hat{\pi}_0\hat{\pi}_1-1}\left( 1-\frac{\hat{n}_1}{n} \right)& -\frac{\hat{\pi}_0}{\hat{\pi}_0\hat{\pi}_1-1}\left( 1-\frac{\hat{n}_1}{n} \right)\\ -\frac{\hat{\pi}_1}{\hat{\pi}_0\hat{\pi}_1-1}\left( 1-\frac{\hat{n}_0}{n} \right)&1+\frac{1}{\hat{\pi}_0\hat{\pi}_1-1}\left( 1-\frac{\hat{n}_0}{n} \right)\end{array}\right]\left[ \begin{array}{c} \hat{\mu}_{1,\AN}\\ \hat{\mu}_{0,\AN} \end{array}\right]=\left[ \begin{array}{c}\frac{\hat{S}_1}{n}+\frac{\hat{T}_1\hat{\pi}_0-\hat{T}_0}{\hat{\pi}_0\hat{\pi}_1-1}\left( 1-\frac{\hat{n}_1}{n} \right)\\ \frac{\hat{S}_0}{n}+\frac{\hat{T}_0\hat{\pi}_1-\hat{T}_1}{\hat{\pi}_0\hat{\pi}_1-1}\left( 1-\frac{\hat{n}_0}{n} \right)\end{array}\right].$$

Finally, solving these equations to recover $\hat{\mu}_{1, \AN}$, $\hat{\mu}_{0, \AN}$, and $\hat{\tau}_{\AN^2}:=\hat{\mu}_{1, \AN}-\hat{\mu}_{0, \AN}$ can be done using any standard matrix inversion technique. We refer to the resulting estimator as $\hat{\tau}_{\AN^2}$ to distinguish it from $\hat{\tau}_{\AN}$.

Somewhat surprisingly, $\hat{\tau}_{\AN^2}$ does not always improve on the MSE of $\hat{\tau}_{\AN}$. The results of a simulation are shown in Figure~\ref{fig:joint-vs-separate}. The normal model simulations suggest that $\hat{\tau}_{\AN^2}$ and $\hat{\tau}_{\AN}$ are preferable in different regimes, but the power law model clearly shows that $\hat{\tau}_{\AN}$ is better there. This suggests that, although $\hat{\tau}_{\AN^2}$ is estimating the optimal values of $\lambda_1, \lambda_0$, the complicated functional form of the data-dependent estimates of the optimal values substantially inflates variance. Although $\hat{\tau}_{\AN}$ is estimating a sub-optimal pair of values $\lambda_1, \lambda_0$, its estimates of those values are lower variance, and ultimately lead to a lower variance estimator.

\begin{figure}[t]
	\centering
	\begin{tabular}{cc}
		\includegraphics[width=0.48\linewidth]{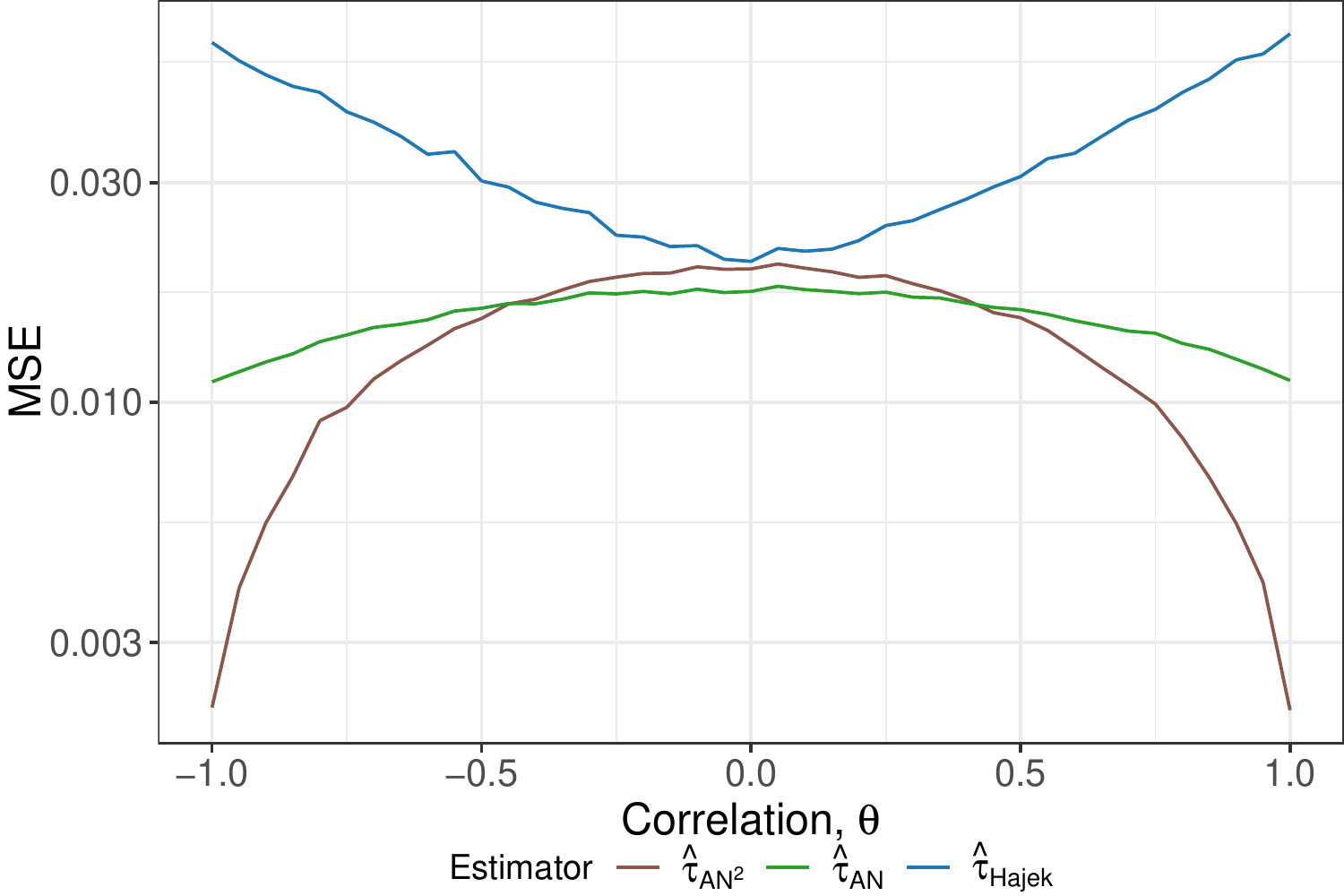}&\includegraphics[width=0.48\linewidth]{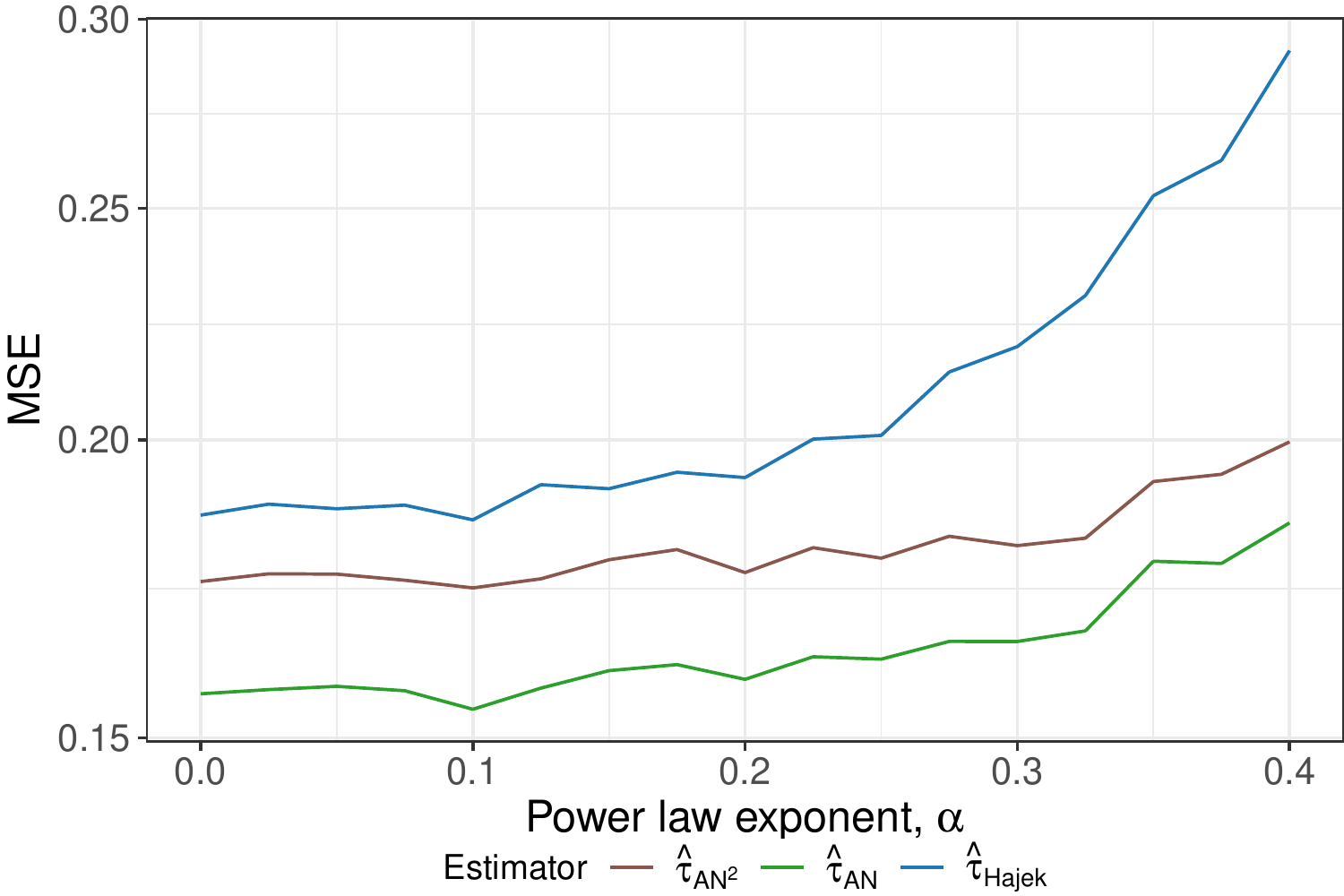}
	\end{tabular}
	\caption{Comparison of $\hat{\tau}_{\AN^2}$ and $\hat{\tau}_{\AN}$ in the normal model with $n=500, \mu=1$ (left) and power law model with $n=500$ (right). In the normal model, $\hat{\tau}_{\AN^2}$ is sometimes better than $\hat{\tau}_{\AN}$, but in the power law model, $\hat{\tau}_{\AN^2}$ is consistently worse.}
	\label{fig:joint-vs-separate}
\end{figure}

\end{document}